\documentclass[lettersize,journal]{IEEEtran}
\usepackage{amsmath,amsfonts,algorithmic,algorithm,array,textcomp,stfloats,url,verbatim,graphicx,cite,amssymb,amsthm,algorithmic,graphicx,textcomp,xcolor,float,tabularx,epstopdf,graphicx,colortbl,blkarray,arydshln,mathtools,nicematrix,physics,algorithm}

\newtheorem{theorem}{Theorem}

\newtheorem{lemma}{Lemma}

\DeclarePairedDelimiter\ceil{\lceil}{\rceil}
\DeclarePairedDelimiter\floor{\lfloor}{\rfloor}

\usepackage{tikz}
\usetikzlibrary{automata,arrows,positioning,calc}

\def\BibTeX{{\rm B\kern-.05em{\sc i\kern-.025em b}\kern-.08em
		T\kern-.1667em\lower.7ex\hbox{E}\kern-.125emX}}

\begin{document}

\title{Minimizing Age of Information \\ with Generate at Will Status Updates \\ and Age-Agnostic Cyclic Scheduling }

\author{Ege~Orkun~Gamgam,~Nail~Akar,~\IEEEmembership{Senior Member,~IEEE},~and~Sennur~Ulukus,~\IEEEmembership{Fellow,~IEEE}
\thanks{E.~O.~Gamgam is currently with Aselsan Inc., Ankara, T\"{u}rkiye (e-mail: eogamgam@aselsan.com.tr). This work is done during the PhD thesis study of E.~O.~Gamgam at Bilkent University.}
\thanks{N.~Akar is with the Electrical and Electronics Engineering Department, Bilkent University, Ankara, 
T\"{u}rkiye 
(e-mail:akar@ee.bilkent.edu.tr). 
Part of this work is done during the sabbatical visit of N.~Akar to University of Maryland, MD, USA, which is supported in part by the Scientific and Technological Research Council of Turkey 
(T\"{u}bitak) 2219-International Postdoctoral Research Fellowship Program.} 
\thanks{S.~Ulukus is with the Department of Electrical and Computer Engineering, University of Maryland, MD, USA (e-mail:ulukus@umd.edu).}}

\maketitle

\begin{abstract}
We study the scheduling problem for a multi-source single-server \emph{generate-at-will} (GAW) status update system with sources having heterogeneous service times and weights, with the goal of minimizing the weighted sum age of information (AoI). 
In particular, we study \emph{age-agnostic} schedulers which rely only on the first two moments of the source service times and they are relatively easier to implement than their age-aware counterparts which make use of the actual realizations of the service times. In particular, we focus on age-agnostic \emph{cyclic schedulers} with $O(1)$ runtime complexity where status updates from multiple sources are scheduled according to a fixed finite transmission pattern.  We first develop an analytical method to obtain the exact average AoI of each source when a transmission pattern is given. Then, we derive the optimum transmission pattern in closed form for the specific case of two sources. 
 For general number of sources, we propose a novel algorithm, called IS (Insertion Search), for constructing transmission patterns, and we show that IS is capable of producing the optimum pattern for two-source systems, and it outperforms other existing age-agnostic schemes, for the case of more than two sources.  
Numerical examples are presented to showcase the effectiveness of the proposed approach.
\end{abstract}


\section{Introduction}
Timely delivery of time-sensitive updates plays a vital role in networked control and remote monitoring systems. In this paper, we consider the one-hop cellular wireless network in Fig.~\ref{fig:system} with the base station (BS),  also called the server, collecting time-sensitive information from $N$ information sources for each of which an associated physical process is sampled with the sample values and their timestamps written into information packets.  In this setting, a source stands for the process/sensor pair.
The focus of this paper is the {generate-at-will} (GAW) pull-based model \cite{lts2015} for which the server is to decide which source process to sample and collect its time-sensitive information from. The alternative to GAW is the random arrival (RA) model which is left outside the scope of this paper. In RA, status update packet generation is governed by the  source, whereas the server is to schedule the randomly arriving packets from the sources. 
GAW can be viewed as the heavy traffic limit of RA when the packet arrival rates from all sources approach infinity and each source keeps only the freshest packet in its buffer.  
In our proposed GAW framework, the server sends a polling message to the scheduled source which subsequently samples its process and generates a new status update packet to be transmitted towards the server in the uplink direction, while avoiding collisions. For GAW timely status update systems, polling multiple access \cite{kadota2021wifresh} is a promising candidate compared to other well-established channel access techniques including random access \cite{chen2022age,yu2024optimizing} which may be more suitable for RA systems in which sources decide when to sample their processes.
For the purposes of energy efficiency, polling messages can be sent in the form of special wake-up radio signals making polling multiple access a viable alternative for very low power users \cite{ghose2018mac,trotta2019bee,shiraishi2022query,shiraishi2024coexistence}.  Fig.~\ref{fig:system} illustrates the IoT data collection system of interest considered in this paper for which 
we assume that information sources use a suitable modulation and coding scheme (MCS) (see \cite{mcs1,mcs2}) depending on their distance from the BS giving rise to heterogeneous (or asymmetric) packet service times characterized by their first two moments. The deployed MCS (along with optional automatic repeat request (ARQ) or hybrid ARQ (HARQ) mechanisms \cite{arq1,arq2}) is assumed to ensure that packet transmission errors are negligible. We leave the investigation of scheduling in the presence of packet errors outside the scope of this paper. We also note that the heterogeneous network model studied in the current paper is different than slotted models with unit service times and heterogeneous packet error probabilities, which are studied extensively in the existing AoI literature \cite{RoyYates__AgeOfInfo_Survey, kosta_etal_survey, kadota_tn18}. We note that the slotted model can also be studied with the proposed general service times model when ARQ/HARQ re-transmissions take place until a status update packet is eventually received successfully at the BS with high probability.

\begin{figure}[t]
	\centering
	\begin{tikzpicture}[scale=0.25]
	\tikzstyle{note} = [rectangle, dashed, draw, fill=white, font=\footnotesize,
	text width=20em, text centered, rounded corners, minimum height=18em]
	\draw[thick] (4.5,8) circle (1.5)  ;
 \draw[thick](4.5,2) circle (1.5);
	\draw[thick](4.5,-6) circle (1.5);
	\draw (4.5,2) node[anchor=center] {{$2$}} ;	
	\draw (4.5,8) node[anchor=center] {{$1$}} ;
	\draw (4.5,-6) node[anchor=center] {{$N$}} ;
	\draw (4.5,-3.3) node[anchor=south] {{$\vdots$}};
	\draw[thick,dashed,<-] (7,2.5) -- (11.5,2.5) ;
        \draw[thick,->] (7,1.5) -- (11.5,1.5) ;
	\draw[thick,<-,dashed] (7,7.5) -- (11.5,4) ;
        \draw[thick,->] (7,6.5) -- (11.5,3) ;
 
	\draw[thick,dashed,<-] (7,-4.5) -- (11.5,1) ;
 \draw[thick,->] (7,-5.5) -- (11.5,0) ;
	\draw[<-, dotted, very thick] (13,-1.2) arc (270:90:3) node[anchor=east] {};
 \draw (15.7,-1.2) node[anchor=center] {\small{\textcolor{black}{scheduler}}};
	\draw[thick](15,2) circle (2);
 
    \filldraw (15,2)  node[anchor=center] {{BS}};	
    \draw[thick,dashed,->] (12,-4) -- (15,-4) ;
    \filldraw (15,-4)  node[anchor=west] {\small{polling message}};
    \draw[thick,->] (12,-5.5) -- (15,-5.5) ;
    \filldraw (15,-5.5)  node[anchor=west] {\small{status update}};
	\end{tikzpicture}
	\caption{Status update packets from $N$ information sources are collected by the base station (BS) employing an age-agnostic scheduler. A polling message is first sent to the scheduled source which in turn samples and transmits its status update packet to the BS using a certain modulation and coding scheme specific to the source.}
	\label{fig:system}
\end{figure}
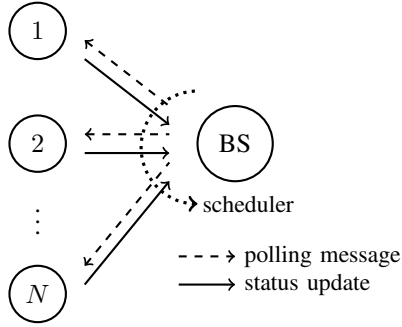

A well-established metric to quantify timeliness is derived from the age of information (AoI) process for source-$n$, $n=1,2,\ldots,N$ which is defined as the random process $\Delta_n(t)= t-g_n(t)$ where $g_n(t)$ is the generation time of the last status update packet received at the destination from source-$n$. Sample paths of the AoI process $\Delta_n(t)$ increase in time with unit slope with abrupt drops at packet reception instances, whereas the peak AoI (PAoI) process is obtained by sampling the AoI process just before the packet reception instants. 
In order to provide system level freshness, the goal of the scheduler in Fig.~\ref{fig:system} is to minimize the weighted sum AoI, or system AoI in short, which is defined as the weighted sum of the mean values of the individual ergodic AoI processes. 
% Here, the weights are used to indicate the relative urgencies of the individual sources.

The scheduling problem with the goal of system AoI minimization is known to be an NP-hard problem in general heterogeneous multi-source settings \cite{he_yuan_ephremides_T-IT18}. However, some special cases are shown to be tractable in \cite{he_yuan_ephremides_T-IT18} along with their optimality conditions. Two main types of schedulers have been studied in the literature, namely, age-aware and age-agnostic schedulers. Age-aware schedulers make use of the instantaneous values of the AoI processes while making scheduling decisions, e.g., max-weight or Whittle-index policies, \cite{kadota_tn18,maatouk_etal_TWC21,kriouile_etal_T-IT22}. On the other hand, age-agnostic schedulers rely only on the a-priori information on the per-source service times and source weights. In our proposed framework, the focus is on the minimizaton of the weighted AoI at the BS. The case of remote monitors not co-located with the BS, i.e., multi-hop scenario, for which the AoI at the monitors may be different than the AoI at the server, is outside the scope of this paper. 
In the setting of this paper, age-aware schemes can also be used and we will compare our proposed age-agnostic schemes with existing age-aware schemes in certain scenarios, whereas we note that age-aware schemes may not be used as effectively when the focus is on the weighted AoI at the remote monitors that are not co-located with the scheduler.

Two types of age-agnostic schedulers stand out depending on whether the scheduler chooses to serve information sources 
according to a fixed pattern of transmissions which repeats itself, called cyclic scheduling or cyclic GAW (C-GAW) \cite{eywa,gau_ciss24},
or
with certain probabilities \cite{akar_gamgam_comlet23}, called probabilistic GAW (P-GAW).
The focus of this paper is on C-GAW due to its superior performance over P-GAW \cite{gau_ciss24} but P-GAW will be used in the numerical examples for benchmarking purposes. 
C-GAW is characterized with a pattern of size $K \geq N$, $P=[P_0,P_1,\ldots,P_{K-1}]$
with $P_i$ being the index of the $i$th source in the transmission pattern, for which the server continuously polls the sources in the following order $P_0,P_1,\ldots,P_{K-1},P_0,P_1,\ldots$. 
% As an example, when $N=2$, assume $P=[1,2,2]$. Then, the server will transmit the source packets continuously in the following order: $1,2,2,1,2,2,1,\ldots$. 
In a symmetric network with identically distributed service times and identical weights across the sources, the optimum transmission pattern is naturally $P=[1,2,\ldots,N]$, by which the server will serve each source on a round-robin (RR) basis. However, as will be shown in the numerical examples, different patterns than that of the RR policy can be obtained in heterogeneous networks that can reduce the system AoI significantly with respect to RR. 
The transmission pattern obtained by the offline cyclic scheduling algorithm is then placed in the memory in a circular buffer data structure, and the runtime complexity of C-GAW is $O(1)$ since it takes one pointer update at a scheduling instant to retrieve the next scheduled source from the circular buffer.
Therefore, C-GAW is advantageous (in terms of runtime implementation) than P-GAW and index-based age-aware scheduling which both have $O(N)$ runtime complexity. 
Moreover, most existing age-aware schedulers are developed for slotted systems with heterogeneous error probabilities. In particular, there is no age-aware scheduler, to the best of our knowledge, proposed for heterogeneous service times (in the first two moments) and weights, which is the setting of the current paper. The focus of this paper is the development of fixed transmission patterns for general number of sources. 

Our contributions are as follows: 
\begin{itemize} 
\item
We develop an analytical method to calculate the mean AoI of an information source for C-GAW for a given transmission pattern using only the first two moments of the service times. We also develop a similar method for the P-GAW benchmark when the transmission probabilities are given. 
\item  We derive the optimum C-GAW scheduler for the case of two sources. We show that the optimum schedule is in the form of transmission of a number of subsequent updates for one of the sources followed by the transmission of a single update for the other source. We obtain the parameters of this policy in closed form, and show that they depend only on the first and second moments of the source service times in addition to the source weights. 
\item Based on the analytical method we have developed, we propose a heuristic search-based cyclic scheduler called IS (Insertion Search) for general number of sources. 
We show that IS produces the optimum pattern for two-source systems whereas it outperforms other existing age-agnostic schemes for the case of larger than two sources.  We also show through numerical examples that IS with $O(1)$ runtime complexity performs close to optimum age-aware schemes with $O(N)$ runtime complexity (for  scenarios when the latter is applicable), especially for small number of users.
% \item We propose additional mechanisms for reusing the proposed transmission patterns developed for GAW systems in certain non-GAW random arrival (RA) scenarios. 
\end{itemize}

The remainder of the paper is organized as follows. Section~\ref{sec:related} overviews the related work. Section~\ref{sec:systemmodel} presents the system model. In Section~\ref{sec:analysis}, we present the analytical method to obtain the system AoI for C-GAW and also P-GAW. Section~\ref{sec:twosources} derives the optimum C-GAW scheduler in closed form for the case of two sources. Section~\ref{sec:Nsources} presents the proposed IS scheduler for general number of sources. In Section~\ref{sec:numerical}, we provide an extensive numerical experimentation for the proposed C-GAW schedulers. Finally, Section~\ref{sec:conclusions} concludes the paper.

\section{Related Work} \label{sec:related}
There has been a surge of interest on AoI modeling and optimization since the concept of AoI was first introduced in a single-source, single-server queuing system setting \cite{kaul_etal_infocom12}. The surveys \cite{RoyYates__AgeOfInfo_Survey} and \cite{kosta_etal_survey} examine in detail the existing work in this field while laying out several open problems. For RA multi-source models, stochastic hybrid systems (SHS) approach is proposed by \cite{Yates__SBR} for obtaining the mean AoI for various buffer management schemes. A closed-form expression for the average age of each source in a globally preemptive M/G/1/1 status update system is studied using the detour flow graph method in \cite{najm2018status}. The SHS approach is also extended to obtain the moment generating function (MGF) and also the higher moments of AoI in various works \cite{yates_mgf, elmagid_dhillon_tit22, moltafet_etal_tcom22, maatuk_isit22}.

A recent approach different than SHS is proposed in \cite{akar_gamgam_comlet23} for finding the distribution of AoI in multi-source probabilistic GAW and also RA servers that are equipped with a single packet buffer, using absorbing Markov chains (AMC). However, for GAW servers, \cite{akar_gamgam_comlet23} requires the use of phase-type service times, which is required when the interest is on the distribution of AoI. When the interest is on the mean AoI only (as in the current paper), then one does not need the distributions of the service times as an input. Hence, the current work is different than \cite{akar_gamgam_comlet23} since the mean AoI values will be shown to be dependent only on the first two moments of the per-source service times for both probabilistic and cyclic GAW schedulers.

For scheduling to minimize AoI, \cite{bedewy_etal_tit21} considers a multi-source system with identically distributed per-source service times, and shows that for any given sampling strategy, the maximum age first (MAF) strategy provides the best age performance among all scheduling strategies. \cite{joo_eryilmaz_ton18} proposes an age-based scheduler that combines age with the inter-arrival times of incoming packets in scheduling decisions to achieve improved information freshness. Although the analytical results are obtained for only heavy-traffic, its numerical results reveal that the proposed algorithm achieves desirable freshness performance for lighter loads as well. \cite{kadota_tn18} and \cite{kadota_tmc21} consider an asymmetric discrete-time wireless network with a BS serving multiple traffic streams using per-source queuing under the assumption of synchronized and random information packet arrivals, respectively, and propose nearly optimal age-based schedulers and age-agnostic randomized schedulers. 
\cite{yang_etal_icassp20} considers an age minimization problem via an age-based scheduling policy which jointly accounts for the staleness of the received parameters and also the instantaneous channel qualities, to improve the running efficiency of federated learning.
\cite{timely_throughput_1} and \cite{timely_throughput_2} study the minimization of AoI subject to timely throughput constraints, whereas \cite{energy_constraint} considers an energy constraint.
\cite{csma} develops distributed CSMA-type protocols for single-hop wireless networks, and \cite{spatially_distributed} considers a spatially distributed random access wireless network for AoI minimization.

Our scheduling work is most closely related to \cite{eywa},\cite{eywa_phDthesis} where a general design framework (Eywa) is proposed to build high-performance cyclic schedulers for AoI-related optimization and decision problems. However, the framework presented in \cite{eywa} is based on a discrete-time setting with deterministic service times (one slot duration) whereas in our work, we study generally distributed service times. Our model is more suitable for physical layer networks for which users can choose an appropriate MCS depending on their distance from the BS, whereas the framework of \cite{eywa} is more suitable when this flexibility is not present. In this regard, we allow general service times to be able to accurately characterize the underlying physical layer. Our work is also related to \cite{AKAR2023109668} which focuses on the minimization of weighted sum PAoI (as opposed to AoI) in GAW status update systems.
However, the second moments of the service times do not play a role in weighted sum PAoI optimization in \cite{AKAR2023109668} for which an optimum solution is known for general number of sources. The current paper is an extension of \cite{gau_ciss24} with an elaborate assessment of the IS algorithm in various scenarios, whereas our previous work \cite{akar_etal_infocom24} focuses on cyclic scheduling algorithms (scalable to thousands of sources) tailored more for deterministic service times and heterogeneous error probabilities with the goal of alignment with the existing literature on AoI scheduling. 

\section{System Model} \label{sec:systemmodel}
We consider the status update system in Fig.~\ref{fig:system} with $N$ information sources indexed by $n=1,2,\ldots,N$ each of which generates samples of an associated random process with the sample values and timestamps written into information packets, which are to be transmitted to the BS which is the destination of the status update system. Service times of source-$n$ packets, denoted by $S_n$, time required for transmission of source-$n$ information packets to the BS, are i.i.d.~and have a general distribution with moment generating function (MGF) denoted by $G_n(s)=\mathbb{E}[e^{s S_n}]$, mean $s_n =\mathbb{E} [S_n] = G_n'(0)$, second moment $q_n = \mathbb{E} [S_n^2]=G_n''(0)$, variance $v_n = q_n - s_n^2$, and squared coefficient of variation $c_n = \frac{v_n}{s_n^2}$. 
In fact, we only need to know the first two moments, but the MGF formalism is used in detailing the analytical method. 
The per-source service times are independent of each other and there are no packet errors.
We also assume that transmission times for polling messages are much shorter than the packet service times and they are neglected in the analysis.

We let $\Delta_n(t)$, $t\geq 0$ (resp. $\Phi_{n,k}$, $k=1,2,\ldots$) denote the AoI process (resp. PAoI process) observed at the BS for source-$n$, $n=1,\ldots,N$. Fig.~\ref{fig:SamplePath1} illustrates the evolution of the AoI/PAoI processes for tagged source-$1$ when C-GAW is used with an example pattern $P=[1,2,3,2]$. We let $\Delta_n$ (resp. $\Phi_n$) denote the steady-state random variable for the associated AoI (resp. PAoI) process observed at the BS for source-$n$, $n=1,\ldots,N$, with their means 
written as time averages due to the ergodicity of both processes,
\begin{align}
 \mathbb{E} [\Delta_n] & =  \lim_{T \rightarrow \infty} \frac{1}{T} \int_{t=0}^{T} \Delta_n(t) \dd{t}, \label{eqn:sourceAoI} \\
 \mathbb{E} [\Phi_n]  & =  \lim_{T \rightarrow \infty} \frac{1}{T} \sum_{k=1}^{T} \Phi_{n,k}.
 \label{eqn:sourcePAoI}
\end{align}
We define the system AoI (resp. system PAoI) as the mean of weighted AoI (resp. weighted PAoI) $\Delta = \sum_{n=1}^N w_n \Delta_n$ (resp. $\Phi= \sum_{n=1}^N w_n \Phi_n$),
\begin{align}
 \mathbb{E} [\Delta] & =  \sum_{n=1}^N w_n \mathbb{E} [\Delta_n],  \label{eqn:systemAoI} \\
 \mathbb{E} [\Phi] & =  \sum_{n=1}^N w_n \mathbb{E} [\Phi_n], \label{eqn:systemPAoI}
\end{align}
where the normalized source weight $w_n$, $n=1,\ldots,N$ (with unit sum) reflects the relative urgency of source-$n$ with respect to others. 

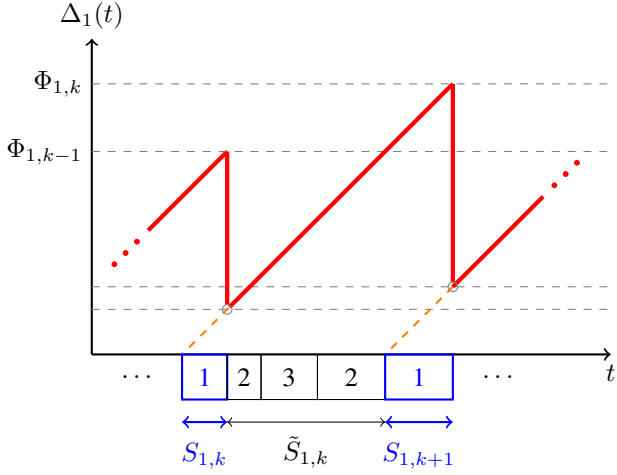
\begin{figure}[tb]
	\centering
	\begin{tikzpicture}[scale=0.30]
	\draw[thick,->] (0,0) -- (23,0) node[anchor=north] {$t$};
	\draw[thick,->] (0,0) -- (0,14) node[anchor=south] {$\Delta_{1}(t)$};
	\draw[ultra thick,red] (2.5,5.5) -- (6,9);
	\filldraw[red] (2,5) circle (3pt);
	\filldraw[red] (1.5,4.5) circle (3pt) ;
	\filldraw[red] (1,4) circle (3pt); 
	\draw (0,9) node[anchor=east] {$\Phi_{1,k-1}$};
	\draw[dashed,gray] (0,9) -- (23,9);
	\draw[dashed,gray] (0,3) -- (23,3);
	\draw[dashed,gray] (0,2) -- (23,2);
%	\draw[dashed,gray] (9,12.5) -- (9,0)  node[anchor=north, thick, black] {$d_{j}$};
%	\draw[dashed,thick,gray] (9,5) -- (4,0)  node[anchor=north, thick, black] {$t_{j}$};
	\draw[dashed,thick,orange] (16,12) -- (4,0);  % node[anchor=north, thick, black] {$t_{j}$};
	\draw[dashed, thick, orange] (16,3) -- (13,0);  % node[anchor=north, thick, black] {$t_{j+1} \quad$};
	\draw[ultra thick,red] (6,9) -- (6,2.2);
	\draw[ultra thick,red] (6.1,2.1) -- (16,12);
	\draw (0,12) node[anchor=east] {$\Phi_{1,k}$};
	\draw[dashed,gray] (0,12) -- (23,12);
	\draw[ultra thick,red] (16,12) -- (16,3.2);
	\draw[ultra thick,red] (16,3) -- (20,7);
	\filldraw[red] (20.5,7.5) circle (3pt);
	\filldraw[red] (21,8) circle (3pt);
	\filldraw[red] (21.5,8.5) circle (3pt); 
	\draw[gray] (6,2) circle (6pt);
	\draw[gray] (16,3) circle (6pt);
 \draw[blue,thick] (4,-2) rectangle (6,0) node[midway,blue]{1};
 \draw[] (6,-2) rectangle (7.5,0) node[midway]{2};
 \draw[] (7.5,-2) rectangle (10,0) node[midway]{3};
% \filldraw[color=blue!60, fill=gray!30, very thick] (10.1,-2) rectangle (11.5,0) node[midway,blue]{1};
 \draw[] (10,-2) rectangle (13,0) node[midway]{2};
 \draw[blue, thick] (13,-2) rectangle (16,0) node[midway,blue]{1};
\node[black,rectangle] at (2,-1) {$\ldots$};
\node[black,rectangle] at (18,-1) {$\ldots$};
\draw[blue,thick,<->] (4,-3) -- (6,-3);
\draw[blue] (5,-3.5) node[anchor=north] {$\footnotesize{S_{1,k}}$};
\draw[black,<->] (6,-3) -- (13,-3);
\draw[black] (9.5,-3.2) node[anchor=north] {$\footnotesize{\tilde{S}_{1,k}}$};
\draw[blue,thick,<->] (13,-3) -- (16,-3);
\draw[blue] (14.5,-3.5) node[anchor=north] {$\footnotesize{S_{1,k+1}}$};
% \draw[red] (22.5,-3.2) node[anchor=north] {$\mathbb{E}[\tilde{S}_1]=\tilde{s}_1, \mathbb{E}[\tilde{S}_1^2]=\tilde{q}_1$};
%\draw[blue] (-7.5,-3) node[anchor=north] {$\mathbb{E}[{S}_1]={s}_1, \mathbb{E}[{S}_1^2]={q}_1$};
\end{tikzpicture}
\caption{Sample path of the AoI/PAoI processes for source-$1$ for C-GAW with $P=[1,2,3,2]$.}
\label{fig:SamplePath1}
\end{figure}

The C-GAW server employs a pattern of size $K \geq N$,
\begin{align}
P & = [P_0, P_1, \ldots, P_{K-1}], \ P_k \in \{1,2,\ldots,N \}.   
\end{align} 
The resulting C-GAW scheduler initiates the sampling and transmission of source-$P_k$ information packet at scheduling instants $k+iK$, $i \in \mathbb{Z^+}$, where $\mathbb{Z^+}$ denotes the set of non-negative integers. The pattern size $K$ cannot be strictly less than $N$ since in that case some user packets would not ever be transmitted leading to unbounded weighted AoI. The specific scenario $K=N$ corresponds to RR scheduling whereas in the $K > N$ case, some users will appear more than once in the pattern, which has advantages in terms of system AoI in relatively heterogeneous scenarios.
The goal of the paper is to find the optimum C-GAW transmission pattern that minimizes the system AoI.  However, P-GAW systems minimizing system AoI and also system PAoI, will be used as benchmark policies for comparison purposes. In particular, the benchmark P-GAW server schedules the information source-$n$ for sampling and transmission with probability $p_n$ with $\sum_{n=1}^N p_n=1$, i.e., work-conserving server, at a scheduling instant. Once the transmission of a source-$n$ packet is completed (which requires a duration of $S_n$), a new scheduling instant is initiated.

\section{Analytical Method} \label{sec:analysis}
Let $S_{n,k}$ denote the service time of the $k$th transmission of source-$n$, for $k \geq 1$, and let $\tilde{S}_{n,k}$ denote the time duration between the end of the $k$th transmission and the beginning of the $(k+1)$st transmission of source-$n$. Note that $\tilde{S}_{n,k}$ is the sum of the service times of all information packets generated from sources other than source-$n$ between two successive transmissions, $k$ and $k+1$, of source-$n$. Common to both types of schedulers, the $k$th cycle, $k \geq 1$, of the AoI process $\Delta_n(t)$ starts from the value $S_{n,k}$ and increases with unit slope for a duration of $\tilde{S}_{n,k} + S_{n,k+1}$ until it hits the value $S_{n,k} + \tilde{S}_{n,k} + S_{n,k+1}$ at which point a new cycle is initiated which will start from the value $S_{n,k+1}$. See Fig.~\ref{fig:SamplePath1} for the evolution of the AoI process for source-$1$ for an example scenario. The duration, and the area under the $\Delta_n(t)$ curve, of the $k$th AoI cycle, denoted by $T_{n,k}$ and $A_{n,k}$, respectively, can be written as,
\begin{align}
T_{n,k} & = \tilde{S}_{n,k} + S_{n,k+1}, \label{onemli1} \\ 
  A_{n,k} & =  S_{n,k} (\tilde{S}_{n,k} + S_{n,k+1}) + \frac{1}{2} (\tilde{S}_{n,k} + S_{n,k+1})^2. \label{onemli2}
\end{align}

Note that the asymptotically stationary discrete-time random process $T_{n,k}$ has a limiting marginal distribution, i.e.,  $T_{n,k} \sim T_n$, i.e., $\lim_{k \rightarrow \infty} \mathbb{P}[T_{n,k} \leq x] = \mathbb{P}[T_{n} \leq x]$, for all $x \geq 0$. Similarly, $A_{n,k} \sim A_n$ and $\tilde{S}_{n,k} \sim \tilde{S}_n$, for the random variables $A_n$ and $\tilde{S}_n$, respectively. 
Let random variable $\tilde{S}_n$ have the MGF $\tilde{G}_n(s)=\mathbb{E}[e^{s\tilde{S}_n}]$, mean $\tilde{s}_n =\mathbb{E} [\tilde{S}_n]$, second moment $\tilde{q}_n = \mathbb{E} [\tilde{S}_n^2]$, and variance $\tilde{v}_n = \tilde{q}_n - \tilde{s}_n^2$ and squared coefficient of variation (scov) $\tilde{c}_n = \frac{\tilde{v}_n}{\tilde{s}_n^2}$.

The mean AoI for source-$n$ $\mathbb{E}[\Delta_n]$ is then written as the following ratio, see e.g., \cite{RoyYates__AgeOfInfo_Survey},
\begin{align}
  \mathbb{E}[\Delta_n]  & = \frac{\mathbb{E}[A_{n}]}{\mathbb{E} [T_n]}, 
  \end{align}
  which then can be written (by taking the ratio of the expectations of the equations \eqref{onemli1} and \eqref{onemli2}) in terms of the parameters $s_n$, $q_n$ (or $v_n$ or $c_n$), $\tilde{s}_n$, and $\tilde{q}_n$ (or $\tilde{v}_n$ or $\tilde{c}_n$), as 
  \begin{align}
 \mathbb{E}[\Delta_n]  & = \frac{2 s_n^2 + 4 s_n \tilde{s}_n + q_n + \tilde{q}_n}{2(s_n + \tilde{s}_n)}, \label{exp_2mom} \\
  & =  \frac{(3s_n + \tilde{s}_n)(s_n + \tilde{s}_n)+ v_n + \tilde{v}_n}{2(s_n + \tilde{s}_n)}, \label{exp_var} \\
  & = \frac{s_n^2 (c_n + 3) +\tilde{s}_n^2(\tilde{c}_n + 1) + 4 s_n \tilde{s}_n }{2(s_n + \tilde{s}_n)}. \label{exp_scov}
\end{align}
On the other hand, the mean PAoI for source-$n$ is written as,
 \begin{align}
 \mathbb{E}[\Phi_n]  & = 2 s_n + \tilde{s}_n. \label{exp_PAoI}
\end{align}

In order to find $\mathbb{E}[\Delta_n]$ (resp. $\mathbb{E}[\Phi_n])$, we need to obtain the first two moments (resp. the first moment) of $\tilde{S}_n$, which is a function of the scheduler. For convenience of notation, we denote by $s_Q$, $v_Q$ and $q_Q$, the mean, variance, and the second moment, respectively, of the sum of the service times needed to transmit any given pattern $Q=[Q_0,Q_1,\ldots,Q_{R-1}]$ of size $R$,
\begin{align}
   s_Q =\sum_{i=0}^{R-1} s_{Q_i}, \ v_Q =\sum_{i=0}^{R-1} v_{Q_i}, \ q_Q  = v_Q + s_Q^2. \label{func_frame}
\end{align}
% how the scheduling actions are made. In the next two sub-sections, we will specialize the general eqn.~(\ref{exp_2mom}) for P-GAW and C-GAW scheduling schemes.

Let us now start with the C-GAW scheduler for a given pattern $P$ with size $K$ for which source-$n$ appears $\alpha_{P,n} > 0$ times in the frame. 
Let $\tilde{P}(n,k)$, $k=0,1,\ldots,\alpha_{P,n}-1,$ denote the sub-pattern obtained by deleting all entries in the original pattern except for the entries between the $k$th and $(k+1)$st $\left( \text{modulo } \alpha_{P,n} \right)$  appearances of source-$n$, excluding the end points. As an example, consider source-1 for a scenario with $N=3$, and
\begin{align}
 K &=7, \ P=[3, 1, 2, 3, 1, 3, 2].    \label{example1}
\end{align}
In this case, $\alpha_P(1)=2$, and we have two sub-patterns for source-1,
\begin{align}
\tilde{P}(1,0) &=[2, 3], \ \tilde{P}(1,1) = [3, 2, 3].    
\end{align}
Recalling the definitions given in \eqref{func_frame}, we can write
\begin{align}
    \tilde{s}_n & = \frac{1}{\alpha_{P,n}} \sum_{k=1}^{\alpha_{P,n}} s_{\tilde{P}(n,k-1)},\label{firstmoment_cgaw} \\
    \tilde{q}_n & = \frac{1}{\alpha_{P,n}} \sum_{k=1}^{\alpha_{P,n}} q_{\tilde{P}(n,k-1)}. \label{secondmoment_cgaw}
\end{align}
At this stage, we obtain $\mathbb{E} [\Delta_n]$ and $\mathbb{E} [\Phi_n]$ for C-GAW by substituting the expression for $\tilde{s}_n$ in \eqref{firstmoment_cgaw} and that for $\tilde{q}_n$ in \eqref{secondmoment_cgaw}, into \eqref{exp_2mom} and \eqref{exp_PAoI}, respectively.

Let us now turn our attention to the benchmark policy P-GAW. 
At a  scheduling instant, we choose to transmit source-$n$ with probability $p_n$ or some other source $m \neq n$ with probability $1-p_n$. Therefore, the number of transmissions required until a single source-$n$ transmission, denoted by $U_n$, is geometrically distributed with parameter $p_n$. 
Between every two successive source-$n$ transmissions, we have $K_n = U_n-1$ transmissions from other sources. Hence, the probability generating function (PGF) of $K_n$ denoted 
by $H_n(z) = \mathbb{E} [z^{K_n}]$ is given by,
\begin{align}
    H_n(z) & = p_n (1 + (1-p_n)z + (1-p_n)^2 z^2 + \cdots) \\  
           & = \frac{p_n}{1 - (1-p_n)z}.
\end{align}

Note that $\tilde{S}_n$ corresponds to the sum of the service times of $K_n$ transmissions from the sources other than source-$n$. Each of the $K_n$ transmissions belongs to source-$m$ with probability $p_m/(1-p_n)$, $m \neq n$. Let $V_n$ denote the service time of any one of these transmissions. Therefore, using the method of collective marks described in \cite{kleinrock_book}, we write the MGF of $\tilde{S}_n$  as follows,
\begin{align}
     \tilde{G}_n(s) & = \frac{p_n}{1 - (1-p_n) \mathbb{E} [e^{s V_n}]}\\
     & =  \frac{p_n}{1 - \sum_{m \neq n} p_m G_m(s)}.
\end{align}
We differentiate $\tilde{G}_n(s)$ twice to obtain $\tilde{s}_n$ and $\tilde{q}_n$,
\begin{align}
    \tilde{s}_n & = \tilde{G}_n'(0) = \frac{\sum_{m \neq n} p_m s_m}{p_n} \label{firstmoment_pgaw} \\ 
    \tilde{q}_n & =\tilde{G}_n''(0) = \frac{\sum_{m \neq n} p_m q_m}{p_n} + \frac{2(\sum_{m \neq n} p_m s_m)^2}{p_n^2}.
    \label{secondmoment_pgaw}
\end{align}
Substituting the expressions \eqref{firstmoment_pgaw} and \eqref{secondmoment_pgaw} into \eqref{exp_2mom} and \eqref{exp_PAoI}, we obtain the mean AoI and mean PAoI, respectively, for a given source-$n$ for P-GAW.

\section{Optimum C-GAW Server for Two Sources} \label{sec:twosources}
In this section, $N$ is fixed to $2$. Let us start with an arbitrary pattern $P$ having size $K \geq 2$ composed of $K_1 \geq 1$ and $K_2 \geq 1$ updates, for source-$1$ and source-$2$, respectively, with $K = K_1+K_2$. 
% Without loss of generality, let us assume that $K_1 \leq K_2$. 
We define a placement vector $r$ (resp. $z$) of size $K_1$ (resp. $K_2$) such that $r_k$, $k=0,1,\ldots,K_1-1$ (resp. $z_k$, $k=0,1,\ldots,K_2-1$) is equal to the number of source-2 (resp. source-1) transmissions that take place between the $k$th and $(k+1)$st modulo $K_1$ (resp. modulo $K_2$) appearances of source-1 (resp. source-2) in the pattern $P$. Note that $r_k$ (resp. $z_k$) corresponds to the length of the sub-pattern $\tilde{P}(1,k)$ (resp. $\tilde{P}(2,k)$).

Without loss of generality, the C-GAW scheduling policy for $N=2$ sources is characterized with the triple $\pi = (K_1,K_2,r) \in \Pi$, $K_1\geq 1$, $K_2 \geq 1$, and $\sum_{k=1}^{K_1} r_{k-1} = K_2$, where $\Pi$ denotes all feasible cyclic scheduling policies. When a policy $\pi$ is employed at the C-GAW server, the system AoI obtained with this policy is defined as $\mathbb{E}[\Delta^\pi] = w_1 \mathbb{E}[\Delta_1^\pi] + w_2 \mathbb{E}[\Delta_2^\pi]$, where the term $\mathbb{E}[\Delta_n ^\pi]$ is the mean AoI for source-$n$ attained under policy $\pi$. An optimum update policy, denoted by $\pi^*$, is one that minimizes the system AoI among all feasible policies $\pi \in \Pi$, and is not necessarily unique. 

The next lemma is crucial for obtaining the optimum policy.

\begin{lemma} \label{lemma:1}
    For a two-source C-GAW server, and for a given scheduling pattern $P$ with size $K$ and number of source-$1$ and source-$2$ updates $K_1$ and $K_2$, such that $K=K_1+K_2$, and $\gamma = K/K_1$, the optimum placement vector for source-$1$ ${r^*} = [ r_0^*,  \ldots, r_{K_1-1}^* ]$ which jointly minimizes the mean AoI for both sources is given as follows. When $\gamma \in \mathbb{Z}$,
    \begin{align} \label{eq:optimalPlacement_equation-1}
    r_{k-1}^* = \gamma -1, \ 1 \leq k \leq K_1,
    \end{align}
    and when $\gamma\notin \mathbb{Z}$,
    \begin{align}
	r_{k-1}^* & =
	\begin{cases}
		\floor{\gamma}-1, & 1 \leq k \leq K_1 \ceil{\gamma}-K, \\
		\floor{\gamma}, & K_1 \ceil{\gamma}-K < k \leq K_1. \\
	\end{cases} \label{eq:optimalPlacement_equation-2}
    \end{align}
\end{lemma}

\begin{proof}
Without loss of generality, let us assume $K_1 \leq K_2$. In order to find $\mathbb{E} [\Delta_1]$, we employ \eqref{exp_2mom}, \eqref{firstmoment_cgaw} and \eqref{func_frame}, for the special case of $N=2$, to first write
\begin{align}
   \tilde{s}_1=\frac{K_2 s_2}{K_1}, \ \tilde{q}_1 = \frac{K_2 v_2}{K_1} + \frac{s_2^2}{K_1} \sum\limits_{k=1}^{K_1} r_{k-1}^2,  
\end{align} 
in order to express the expected AoI of the first source $\mathbb{E} [\Delta_1]$ as,
\begin{align}
 \mathbb{E} [\Delta_1]  \!=\! \frac{K_1 (2s_1^2\!+\!q_1) + K_2(4 s_1 s_2\!+\!q_2\!-\!s_2^2) + s_2^2 \sum\limits_{k=1}^{K_1} r_{k-1}^2}
	{2(K_1 s_1 + K_2 s_2)}. \label{eq:optim_1}
\end{align}
To minimize \eqref{eq:optim_1}, it is sufficient to minimize the expression $\sum_{k=1}^{K_1} r_{k-1}^2$ under the constraint $\sum_{k=1}^{K_1} r_{k-1} = K_2$, one solution to which is given in \eqref{eq:optimalPlacement_equation-1}-\eqref{eq:optimalPlacement_equation-2}. Similarly, we can write the expected AoI of the second source
$\mathbb{E} [\Delta_2]$ as
\begin{align}
 \mathbb{E} [\Delta_2]  = \frac{K_2 (2s_2^2\!+\!q_2) + K_1(4 s_1 s_2\!+\!q_1\!-\!s_1^2) + s_1^2 \sum\limits_{k=1}^{K_2} z_{k-1}^2}
	{2(K_1 s_1 + K_2 s_2)}. \label{eq:optim_11}
\end{align}
Now, to minimize \eqref{eq:optim_11}, it is also sufficient to minimize the expression $\sum_{k=1}^{K_2} z_{k-1}^2$ under the constraint $\sum_{k=1}^{K_2} z_{k-1} = K_1$. Since $K_1 \leq K_2$ and $K_1+K_2=K$, we have $K/K_2 \leq 2$. Then, $\mathbb{E} [\Delta_2]$ is minimized if the entries in its placement vector are either 0 or 1 since $K/K_2 \leq 2$ and the number of 1s should be $K_1$. Obviously, such selection of placement vector $z$ for source-$2$ is already achieved with the selection of placement vector $r$ for source-$1$ as given in \eqref{eq:optimalPlacement_equation-1}-\eqref{eq:optimalPlacement_equation-2}. Thus, the selection of placement vector for source-$1$ in \eqref{eq:optimalPlacement_equation-1}-\eqref{eq:optimalPlacement_equation-2} simultaneously minimizes the expected ages of both sources, completing the proof.
\end{proof}

As an illustrative example for Lemma~\ref{lemma:1}, when the size is fixed to $K=7$ and the number of source-$1$ updates is fixed to $K_1=3$, the optimal policy is obtained as $\left(7, 3, [1,1,2]\right)$, whose transmission pattern is given as 
\begin{align}
P &= [ 1, 2, 1, 2, 1, 2, 2 ].  
\end{align}

From Lemma~\ref{lemma:1}, we focus only on policies $\pi = (K_1,K_2) \in \Pi^\prime, K_n \geq 1, n=1,2,$ that use the optimum placement vector, which is equivalent to the policy $\pi=(K_1, K_2, r)$ where the placement vector $r$ is employed according to Lemma~\ref{lemma:1}. 

The next lemma further narrows down the candidate set of optimum policies.

\begin{lemma} \label{theorem:optimalTwoSetPolicies}
    For a two-source C-GAW server, the optimum policy is either in the form $\pi = (K_1,1)$ or $(1,K_2)$.	
\end{lemma}

\begin{proof}
    Let us be given a policy $\pi =(K_1,K_2) \in \Pi^\prime$ with $K=K_1+K_2$, $K_1 \leq K_2$ and $\gamma=K/K_1$ with the placement vector $r$ obtained according to Lemma~\ref{lemma:1}. When $\gamma \in \mathbb{Z}$, then the policy $(K_1,K_2)$ is already equivalent to the policy $(1,\gamma-1)$. 
    
    Now, consider the case 
    $\gamma \notin \mathbb{Z}$. Let us define two auxiliary policies $\pi_a = \left(1,\floor{\gamma} - 1\right)$ and $\pi_b = \left(1,\floor{\gamma}\right)$. The system AoI obtained with the policy $\pi$ can be written as,
    \begin{align}
	\mathbb{E} [\Delta^{\pi}] 	&= \frac
	{t_a \mathbb{E} [\Delta^{\pi_a}] + t_b \mathbb{E} [\Delta^{\pi_b}]}
	{t_a + t_b}, \label{eq:aoi_decom_a}
    \end{align}
    where
    \begin{align}
     t_a &= \left(s_2 (\floor{\gamma}-1)+s_1\right) \left(K_1\ceil{\gamma}-K\right), \\
	t_b &= \left(s_2 (\ceil{\gamma}-1)+s_1\right) \left(K_1-K_1\ceil{\gamma}+K\right).
    \end{align}
    Note that \eqref{eq:aoi_decom_a} immediately reveals that either $\pi_a$ or $\pi_b$ should give a lesser system AoI than the original policy $\pi=(K_1,K_2)$ from which we conclude that the optimum policy must be in the form $(1,K_2)$. When $K_1 > K_2$, one can swap the sources and follow the same lines of the above proof to show that the optimum policy can alternatively be in the form $(K_1,1)$. 
\end{proof}

Let us first focus only on the policies of the form $\pi=(K_1,1)$ and apply convex optimization to find the best value of $K_1$, denoted by $K_1^*$, that minimizes the system AoI.  For this purpose, the mean AoI for each of the two sources obtained with policy $(K_1,1)$ are first written in terms of the parameter $K_1$ using \eqref{exp_2mom} (see also \eqref{eq:optim_1} and \eqref{eq:optim_11}) as,
\begin{align}
    \mathbb{E} [\Delta_1] &= \frac{\ K_1 (2s_1^2 + q_1) + 4 s_1 s_2 + q_2}{2(K_1 s_1 + s_2)}, \label{nail31} \\ %\frac{4s_1^2 + 4 \Theta s_1 s_2 + \Theta s_2^2 + \Theta^2 s_2^2}{2(\Theta s_2 + s_1)}, \\
    \mathbb{E} [\Delta_2] &= \frac{K_1^2 s_1^2 + K_1 (4 s_1 s_2 + q_1- s_1^2) + 2 s_2^2 + q_2}{2(K_1 s_1 + s_2)}. \label{nail32}  %\frac{4 \Theta s_2^2 + 4 s_1 s_2 + 2 s_1^2}{2(\Theta s_2 + s_1)}.
\end{align}
Relaxing the integer nature of $K_1$, we write the first and second derivatives of the system AoI as a function of $x\in \mathbb{R}$ after replacing each occurrence of $K_1$ in the system AoI expression by the real-valued variable $x$,  
\begin{align}
    \dv{x} \mathbb{E}[\Delta] &=  \frac{w_2 s_1}{2} - \frac{w_2 s_1 \psi_1}{2(s_2 + x s_1)^2}, \label{Step3_firstOrderDer} \\
    \dv[2]{x} \mathbb{E}[\Delta]&=  \frac{w_2 s_1^2 \psi_1}{(s_2 + x s_1)^3}, \label{Step3_secOrderDer}
\end{align}
where
\begin{align}
 \psi_1 = \frac{1}{s_1 w_2} (s_1 q_2 -q_1s_2+ (w_1 +1)s_1^2 s_2-w_2 s_2^2s_1). \label{psi_1} 
\end{align}

We note that if $\psi_1 < 0$, then from \eqref{Step3_firstOrderDer}, the function $\mathbb{E}[\Delta]$ is a monotonically increasing function of $x$  in which case $K_1^*=1$. On the other hand, if $\psi_1$ is positive, then from \eqref{Step3_secOrderDer}, the function $\mathbb{E}[\Delta]$ is a convex function of $x$, and using the KKT optimality conditions \cite{boyd2004convex}, we find the value of the parameter $x$, denoted by $x^*$, for which the expression in \eqref{Step3_firstOrderDer} when evaluated at $x^*$ becomes zero. Consequently, we obtain
\begin{align}
x^*= \frac{\sqrt{\psi_1}-s_2}{s_1}. \label{eq:optimalPolicyParameter_2}
\end{align}
If $x^*\leq 1$ which occurs when $\psi_1 \leq (s_1+s_2)^2$, then $K_1^*=1$. Otherwise, if $x^* \in \mathbb{Z}$, then $K_1^*=x^*$ and if $x^* \notin \mathbb{Z}$ then $K_1^*$ is either $\floor{x^*}$ or $\ceil{x^*}$ depending on which of the two choices yields a lower system AoI. 

We now repeat the same analysis for policies in the form of $\pi=(1,K_2)$ and obtain $K_2^*$, that minimizes the system AoI for all such policies. 
For this purpose, we write the mean AoI for the two sources obtained with the policy $(1,K_2)$ using \eqref{exp_2mom} (see also \eqref{eq:optim_1} and \eqref{eq:optim_11}) as,
\begin{align}
    \mathbb{E} [\Delta_1] &= \frac{K_2^2 s_2^2 + K_2 (4 s_1 s_2 + q_2- s_2^2) + 2 s_1^2 + q_1}{2(K_2 s_2 + s_1)}, \label{nail33} \\
    \mathbb{E} [\Delta_2] &= \frac{\ K_2 (2s_2^2 + q_2) + 4 s_1 s_2 + q_1}{2(K_2 s_2 + s_1)}, \label{nail34}  
\end{align}
Relaxing the integer nature of $K_2$, we write the first and second derivatives of the system AoI as a function of $y\in \mathbb{R}$ upon replacing each occurrence of $K_2$ in the system AoI expression by the real-valued variable $y$,  
\begin{align}
    \dv{y} \mathbb{E}[\Delta] &=  \frac{w_1 s_2}{2} - \frac{w_1 s_2 \psi_2}{2(s_1 + y s_2)^2}, \label{Step3_firstOrderDer2} \\
    \dv[2]{y} \mathbb{E}[\Delta]&=  \frac{w_1 s_2^2 \psi_2}{(s_1 + y s_2)^3}, \label{Step3_secOrderDer2}
\end{align}
where
\begin{align}
 \psi_2 = \frac{1}{s_2 w_1} ( s_2 q_1 -q_2s_1+ (w_2 +1)s_2^2 s_1-w_1 s_1^2 s_2 ). \label{psi_2} 
\end{align}
Similar to the previous analysis, if $\psi_2 < 0$, from \eqref{Step3_firstOrderDer2}, the function $\mathbb{E}[\Delta]$ is a monotonically increasing function of $y$ in which case $K_2^*=1$. If $\psi_2$ is positive, then from \eqref{Step3_secOrderDer2}, the function $\mathbb{E}[\Delta]$ is a convex function of $y$. Consequently, we obtain the value of the parameter $y$, denoted by $y^*$, for which the expression \eqref{Step3_firstOrderDer2} when evaluated at $y^*$ becomes zero, 
\begin{align}
y^*= \frac{\sqrt{\psi_2}-s_2}{s_1}. \label{eq:optimalPolicyParameter_22}
\end{align}
If $y^*\leq 1$ which occurs when $\psi_2 \leq (s_1+s_2)^2$, then $K_2^*=1$. Otherwise, if $y^* \in \mathbb{Z}$, then $K_2^*=y^*$ and when $y^* \notin \mathbb{Z}$, then $K_2^*$ is either $\floor{y^*}$ or $\ceil{y^*}$ depending on which of the two choices yields a lower system AoI. Finally, the optimum policy is either $(K_1^*,1)$ or $(1,K_2^*)$, depending on which one gives rise to a lower system AoI.

Theorem~\ref{thm:thm1} below provides a closed-form expression for the optimum C-GAW server for the case of two sources.

\begin{theorem} \label{thm:thm1}
    Consider the system AoI minimization problem for the two-source C-GAW server with the expressions for the parameters $\psi_1$, $x^*$, $\psi_2$, and $y^*$ given in equations \eqref{psi_1}, \eqref{eq:optimalPolicyParameter_2}, \eqref{psi_2}, and \eqref{eq:optimalPolicyParameter_22}, respectively. 
    The inequality $\psi_n >  (s_1+s_2)^2$ cannot hold for the two sources simultaneously. If $\psi_n \leq  (s_1+s_2)^2$, $n=1, 2$, then the round robin policy $(1,1)$ is the optimum policy. If $\psi_1 > (s_1+s_2)^2$, then $(1,\floor{x^*})$ or $(1,\ceil{x^*})$ is the optimum policy depending on which one results in a lower system AoI using \eqref{nail31} and \eqref{nail32}. Similarly, if $\psi_2 > (s_1+s_2)^2$, then $(\floor{y^*},1)$ or $(\ceil{y^*},1)$ is the optimum policy depending on which one results in a lower system AoI using \eqref{nail33} and \eqref{nail34}.  
\end{theorem}

\begin{proof}
  We first prove by contradiction that $\psi_n >  (s_1+s_2)^2$ cannot hold for both values of $n$. Let us first assume, 
  \begin{align}
   \psi_n & = a_n (s_1+s_2)^2, \label{nail35}
  \end{align} 
  where $a_n >1$ for $n=1,2$. It is not difficult to show using \eqref{psi_1} and \eqref{psi_2} that the following holds,
  \begin{align}
    \psi_1 w_2 s_1 + \psi_2 w_1 s_2 & = s_1 s_2 (s_1+s_2). \label{nail36} 
  \end{align}
  Substituting \eqref{nail35} in the expression in \eqref{nail36}, we obtain
  \begin{align}
      \frac{a_2 (s_1 + s_2)}{s_1} w_1 + \frac{a_1 (s_1 + s_2)}{s_2} w_2 & = 1, \label{nail37}
  \end{align}
  which cannot hold since $w_1 + w_2=1$, and therefore, the left hand side of \eqref{nail37} must be strictly larger than $1$ contradicting the assumption \eqref{nail35}. The rest of the proof simply follows from the derivations from \eqref{nail31} to \eqref{eq:optimalPolicyParameter_22}.
\end{proof}

From \eqref{psi_1} and \eqref{psi_2}, it is evident that as $w_1$ (resp. $w_2$) increases, $\psi_1$  (resp. $\psi_2$) increases monotonically, while $\psi_2$ (resp. $\psi_1$) decreases monotonically. When we combine this finding with Lemma~\ref{theorem:optimalTwoSetPolicies}, we observe that if the optimal policy takes the form $(K_1, 1)$, further increasing $w_1$ leads to a higher $K_1$, i.e., more transmissions for source-$1$ in the transmission pattern of the optimal policy. Conversely, if $w_1$ decreases, $K_1$ also decreases until it reaches one, overlapping with the round robin (RR) policy. Further decrease of $w_1$ results in the optimal policy taking the form $(1, K_2)$ with increasing $K_2$.

\section{Insertion Search (IS) Scheduling Algorithm} \label{sec:Nsources}
In the previous section, we obtained the optimum cyclic schedule for $N=2$ sources for C-GAW. Optimum scheduling for general $N>2$ sources for C-GAW remains an open problem. In this section, we propose a heuristic search-based scheduler called Insertion Search (IS) for general number of sources. For the description of the IS algorithm, we first define the insertion operator $g(P,n,k)$, for $1\leq n \leq N$, $0 \leq k < K$, which maps the original frame $P$ with size $K$, to a new pattern $P'$ with size $K+1$, which is obtained by inserting a source-$n$ transmission before the $k$th element of $P$. As an example, let $P$ be given as in \eqref{example1}. Then, $g(P,1,6)$ is obtained by inserting a source-1 transmission just before the 6th entry of $P$,
\begin{align}
g(P,1,6) & =[ 3, 1, 2, 3, 1, 3, 1, 2].    \label{example2}
\end{align}
Given this definition, we first start with the round robin (RR) policy (sources are served one after the other, in a repeating sequence) as our initial policy which results in an initial pattern $P$ such that $P_i = i+1$, $0 \leq i \leq N-1$ with pattern size $N$ which is the only feasible cyclic schedule with pattern size less than or equal to $N$. This stems from the observation that patterns with size strictly less than $N$ lead to unbounded system AoI, and the RR policy is the only policy with pattern size $N$. 
Subsequently, at each iteration of our proposed algorithm, our goal is to find a pattern $P'$ with size $K+1$, using the pattern $P$ with size $K$, through the use of insertion operations only.
At this stage, we find the best such pattern $P'$ in terms of system AoI by searching through each of the sources and their insertion position. This step requires the computation of system AoI for less than $NK$ patterns since the patterns $g(P,n,k)$ and $g(P,n,k+1)$ (modulo $K$) are equivalent  when $P_k =n$. 
% If we can reduce the system AoI by this method by increasing the pattern size, we repeat the process. 
Particularly, this process is terminated if the maximum size $K_{max}$ is reached or the system AoI cannot be improved for $y$ consecutive pattern sizes, where the integer $y$ is an input to the IS algorithm. For the special case of $y=1$, the process is terminated if we cannot reduce the system AoI in a given iteration. In the case of general $y$, although we may not improve the system AoI by increasing the pattern size by one, we continue the insertion search process for a total of $y$ steps. The pseudo-code for IS is given in Alg.~\ref{alg:WFS_pseudo} that outputs a pattern $P^*$ of size $K^*$ which can be at most $K_{max}$, where we use the notation $SAoI(P)$ to denote the system AoI of pattern $P$ which can be calculated by following the procedure detailed in Section~\ref{sec:analysis}. 
In Alg.~\ref{alg:WFS_pseudo}, at the end of an iteration involving a specific pattern size, $CSAoI$ stands for the minimum system AoI attained with the use of the pattern $\tilde{P}$, obtained for the current pattern size, whereas $P^*$ holds the pattern that yields the minimum system AoI, denoted by $SAoI$,  for all pattern sizes tried until the current pattern size. 
The parameter $s$ holds the number of consecutive frame sizes for which the system AoI cannot be improved 
and the while loop is terminated when $s=y$ or the pattern size becomes $K_{max}$.
$P^*$ is the output pattern of Alg.~\ref{alg:WFS_pseudo} and $K^*$ is the output pattern size, which are returned at the end of the while loop.

\begin{algorithm}[tb]
	\caption{Pseudo-code for IS}
	\begin{algorithmic}
    \renewcommand{\algorithmicrequire}{\textbf{Input:}}
    \renewcommand{\algorithmicensure}{\textbf{Output:}}
		\REQUIRE $N$, $w_n$, $s_n$, $q_n$, $1 \leq n \leq N$, $K_{max}$, $y$;
        \ENSURE $P^*$, $K^*$;
		\STATE $K \gets N$;
        \STATE $K^* \gets K$;
        \STATE $P \gets [1,2,\ldots,N] $;
        \STATE $P^* \gets P$;
        \STATE $\tilde{P} \gets P$;
        \STATE  $SAoI \gets SAoI(P)$;
        %\STATE  $C \gets 1$;   
        \STATE  $s \gets 0$;   
		%\STATE $C \leftarrow 0$, $\beta = (1,0)$; %$p_{n',n^*} \leftarrow 1$, $p_{n^*,n'} \leftarrow 0$;	 
        \WHILE{$s < y$ and $K < K_{max}$}     
            \STATE $CSAoI \gets \infty $;
         	\STATE $K \gets K+1$;
          \STATE $s \gets s + 1$;
                \STATE $P \gets \tilde{P}$;
                \FOR{$n=1$ \TO $N$} 
                    \FOR {$k=0$ \TO $K-1$}
                        \IF {$P_k \neq n$}
                            \STATE $P' \gets g(P,n,k)$;
                            \STATE $SAoI' \gets SAoI(P')$;
                    
                            \IF {$SAoI' < CSAoI$} 
                                \STATE $CSAoI \gets SAoI'$;
                                \STATE $\tilde{P} \gets P'$;
                            \ENDIF
                            
                            \IF {$SAoI' < SAoI$} 
                                \STATE {$P^* \gets P'$};
                                \STATE {$\tilde{P} \gets P'$};
                                \STATE $SAoI \gets SAoI'$;
                                \STATE $K^* \gets K$;
                                % \STATE $C \gets 0$;
                                \STATE $s \gets 0$;
                            \ENDIF
                        
                        \ENDIF
                    \ENDFOR
                \ENDFOR
            %     \STATE $P \gets P^*$;
    		
            % \IF {$K = K_{max}$}
            %     \STATE $extra \gets y$;
            % \ELSIF {$C = 0$} 
            %     \STATE $extra \gets extra + 1$;
            %     \STATE $P \gets continued\_P$;
            % \ENDIF 
        \ENDWHILE
	\end{algorithmic}
	\label{alg:WFS_pseudo}
\end{algorithm}

Note that IS provides the optimum solution for $N=2$ from Theorem~\ref{thm:thm1}. To see this, if the optimum solution is $(1,1)$ then IS will start from the pattern $[1,2]$ and will choose not to insert a source since the system AoI would be higher upon an insertion. Otherwise, if the solution is in the form $(K_1^*,1)$ (resp. $(1,K_2^*))$, then from Theorem~\ref{thm:thm1} and Lemma~\ref{lemma:1}, it will never be possible to insert a source-2 (resp. source-1) transmission in order to reduce the system AoI at any step of the IS algorithm. Hence, IS will eventually find the optimum solution for $N=2$ using linear search rather than the closed form expression given in Theorem~\ref{thm:thm1}. The advantage in using IS stems from its applicability to general $N > 2$ scenarios. 

Recall that the runtime complexity of IS is $O(1)$. However, the complexity of Alg.~\ref{alg:WFS_pseudo} involving the offline pattern construction phase is crucial in determining the size of the networks IS can handle. 
For this purpose, note that finding the system AoI of a feasible pattern requires $O(NK)$ floating point operations for a pattern with size $K$ and $N$ sources, from the equations \eqref{exp_2mom}, \eqref{func_frame}, \eqref{firstmoment_cgaw} and \eqref{secondmoment_cgaw}. On the other hand, the IS algorithm needs to find the system AoI for each candidate source for insertion and the associated insertion point. Therefore, the IS algorithm requires at most $O(N^2K^2)$ operations when the IS algorithm is at the phase of switching from size $K$ to $K+1$. Let $K^*$ denote the optimum  pattern size found by the IS algorithm when $K_{max} \rightarrow \infty$. Therefore, the overall offline computational complexity of IS is $O(N^2 \bar{K}^3)$ where $\bar{K} = \min(K^*,K_{max})$, which is to be compared against an Exhaustive Search (ES) algorithm which requires a selection from $N^K$ patterns for finding the minimum system AoI pattern with size $K$. 
However, note that since $\bar{K}>N$, IS has a computational complexity exceeding $O(N^5)$ which may limit its use in massive connectivity scenarios involving thousands of sources. 

\section{Numerical Examples} \label{sec:numerical}
\subsection{C-GAW with Two Sources}
In this subsection, we focus on the case $N=2$ for which we obtained the optimum cyclic scheduler in Section~\ref{sec:twosources} and in Theorem~\ref{thm:thm1}, denoted by C-GAW$^*$. As benchmark policies, we compare  C-GAW$^*$ against RR and the optimum P-GAW policy, denoted by P-GAW$^*$, for which 
%We note that for P-GAW, KKT conditions result in a quartic equation because of which a closed-form expression for $p_1$ does not appear to be possible as in the case of C-GAW. 
we propose to solve the system AoI using the technique described in Section~\ref{sec:analysis} and employ a one-dimensional exhaustive search over the transmission probability $p_1$ that minimizes the system AoI.

In the first example, we consider the GAW model with exponentially distributed service times, i.e., $q_n = 2 s_n^2$, $n=1,2,$ for the case $s_1=5$, $w_1=0.8$ and $w_2=0.2$. The system AoI is depicted in Fig.~\ref{figure:Example_1}(a) as a function of $s_2$  for the three cases: P-GAW$^*$, RR, C-GAW$^*$. Fig.~\ref{figure:Example_1}(a) reveals that the system AoI obtained with RR is lower than that of the optimum probabilistic P-GAW$^*$ server for a wide range of $s_2$ values despite the fact that RR server does not take into consideration the individual source weights or service time statistics. Moreover, the system AoI obtained with the C-GAW$^*$ server is consistently lower than that of the other two servers and this performance gap also grows with increasing $s_2$. We conclude from this example that it is beneficial to use cyclic scheduling when compared to probabilistic and RR scheduling for weighted sum AoI minimization.  When C-GAW$^*$ deviates from RR, we let $n'$ denote the source which requires $K_{n'}^* > 1$ transmissions for the C-GAW$^*$ server, whereas when they are equal, we fix $n'=1$ and  $K_{n'}^* = 1$. The parameter $p_1^*$ for P-GAW$^*$ and the parameter $K_{n'}^*$ for C-GAW$^*$ are depicted in Fig.~\ref{figure:Example_1}(b) which demonstrates that we need to have increasingly more source-$1$ appearances in the optimum transmission pattern as the channel quality of source-$2$ worsens.

 \begin{figure*}[t]
 %\hspace*{1cm}
	\centering
  \includegraphics[width=0.8\textwidth]{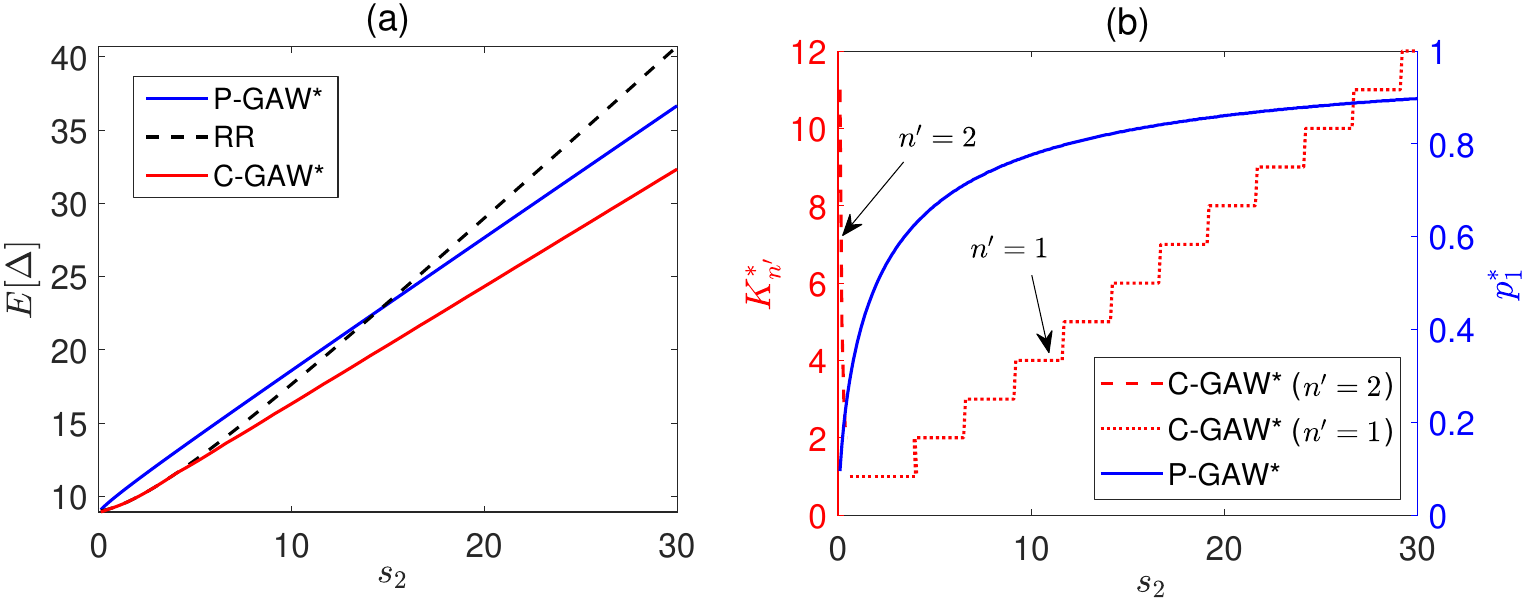}
	\caption{(a) System AoI $\mathbb{E} [\Delta]$ (b) $p_1^*$ for P-GAW$^*$ and $K_{n'}^*$ for C-GAW$^*$, depicted as a function of $s_2$ for exponentially distributed service times when $w_1 = 0.8, w_2=0.2, s_1=5.$}
	\label{figure:Example_1}
\end{figure*}

\begin{figure*}[t]
  \centering
  \includegraphics[width=0.8\textwidth]{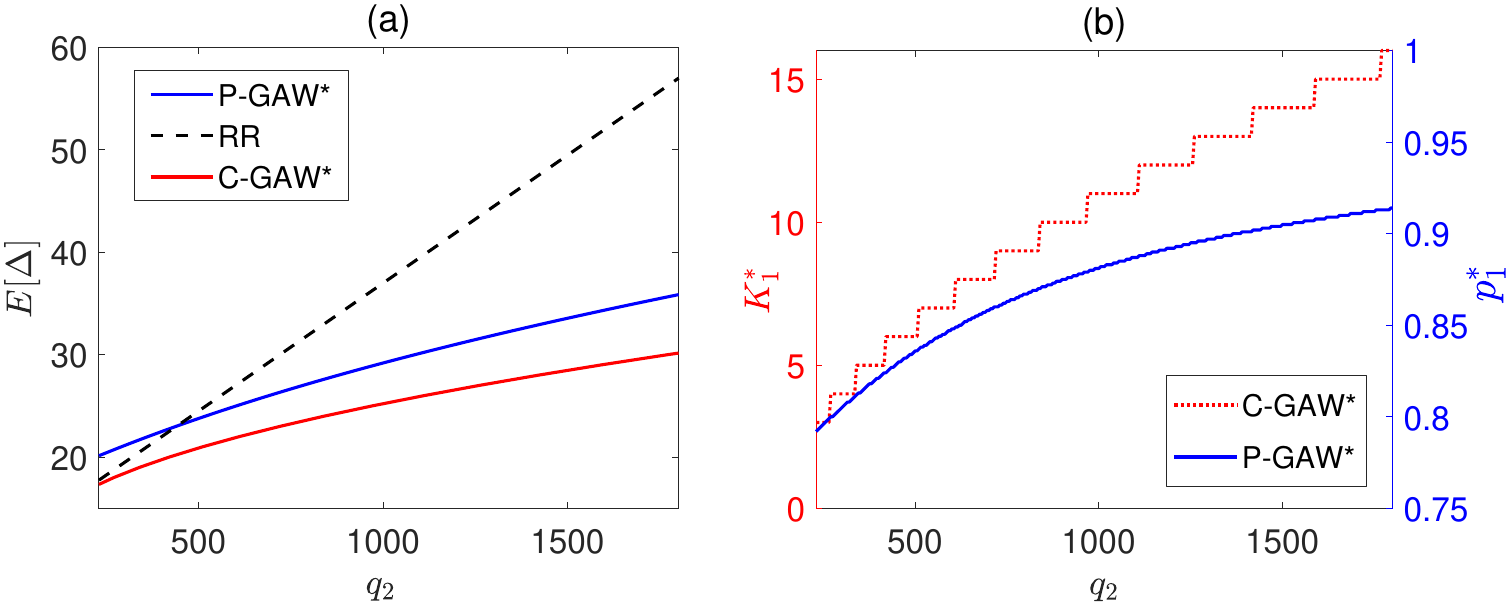}
  \caption{(a) System AoI $\mathbb{E} [\Delta]$ (b) $p_1^*$ for P-GAW$^*$ and $K_1^*$ for C-GAW$^*$ ($n'=1$ for this example in all cases), depicted as a function of $q_2 \geq 225$, when $w_1 = 0.8, w_2=0.2, s_1=5, s_2=15$. }
  \label{figure:Example_2}
\end{figure*}

In the second example, we use the same weights as before and exponentially distributed service times for source-1 with $s_1 = 5$. However, the service times for source-$2$ are now generally distributed with mean $s_2 =15$. The system AoI is plotted in Fig.~\ref{figure:Example_2}(a) for the three cases: P-GAW$^*$, RR, C-GAW$^*$, and the parameters $p_1^*$ and $K_{1}^*$ are plotted in  Fig.~\ref{figure:Example_2}(b), as a function of the second moment of the source-2 service time, $q_2$, for all values of which $n'=1$. We observe that the C-GAW$^*$ server progressively favors source-1, i.e., $K_{1}^*$ is increased, for increasing $q_2$ values. We also observe the adverse effect of the second moment of service times on system AoI and also its impact on the optimum policy. Moreover, C-GAW$^*$ server consistently outperforms the other two and the performance gap increases with increased $q_2$ values. This example illustrates the importance of second moments of the service times on the AoI performance in addition to the first moments.

\subsection{Comparison of IS with Existing Age-Agnostic Schedulers}
% In this example, we fix $N=3$ and employ uniform weights, i.e., $w_n=1/3, \forall n$. The parameters $s_1$ and $s_2$ are fixed to $2$ and $5$, respectively. Then, the mean service time $s_3$ is varied to obtain the system AoI using three policies, namely RR, P-GAW$^*$ and IS, the latter one implementing Algorithm~\ref{alg:WFS_pseudo} with $K_{max} \rightarrow \infty$, which is depicted in Fig.~\ref{fig:IS} for three different scenarios: a) deterministic service times, i.e., $c_n=0$ for all sources, b) exponentially distributed service times, i.e., $c_n=1$ for all sources, c) varying squared coefficient of variation across the three sources with the choice of $c_1=0$, $c_2=1$, $c_3=5$. In all the scenarios, IS outperforms the P-GAW$^*$ and RR policies as in the $N=2$ case when IS is equivalent to the optimum cyclic scheduler C-GAW$^*$. The RR scheduler performs closely to IS in terms of system AoI for deterministic service times even when the mean service times are different for the information sources. 
% However, the gap between IS and RR is larger when the service times are not deterministic but exponentially distributed.
% In Fig.~\ref{fig:IS}(c), the variability in service times is different for each source and in this case, the system AoI is penalized severely when RR is used especially for larger values of $s_3$. In all the cases we have studied, IS outperformed P-GAW$^*$ consistently.
In this subsection, we turn our attention to the general $N>2$ case for which the optimum solution is not available for general heterogeneous scenarios. In the first example, we study the impact of heterogeneity in mean service times and also in weights, on the performance of the IS algorithm. For this purpose, we investigate the system AoI as a function of the iteration index $K$ for the IS algorithm, starting from $K=N$, i.e., RR policy.  We fix the IS algorithm parameter $y=1$  which subsequently produces the pattern size $K^*$ above which a performance improvement is not observed by increasing the pattern size by one. We fix $K_{max} \rightarrow \infty$ and $N=50$ for this numerical example. We use five different mean service times representative of five different MCS values in the cellular network. In particular, in scenario A, $s_n \in \{ 1,2,3,4,5 \}$, and in the more heterogeneous scenario B, $s_n \in \{ 1,2,4,8,16 \}$. When $n \mod 5 =i$, then $s_n$ takes the $i$th value of the corresponding set for $0\leq i < 5$. The scov parameter is fixed for all sources, i.e., $c_n=c, \ 1 \leq n \leq N$, and two different choices of the parameter $c=0,1$ are studied. We assume a \emph{linearly spaced weights} model where $w_i=w_{i-1}+\delta, \ 2 \leq i \leq N, \  \delta>0$ with the ratio of the maximum weight to the minimum weight, $r=\frac{w_N}{w_1}$, is fixed to either $5$ or $50$. Larger values of the ratio parameter $r$ are indicative of increased heterogeneity in the weights. The results are given in Figs.~\ref{fig:IS1}(a) and \ref{fig:IS1}(b), for scenarios A and B, respectively. In all cases, we have observed significant reduction in the system AoI using IS with respect to the RR policy which is known to be optimum for the homogeneous network scenario.  However, when heterogeneity is lower in weights, i.e., $r=5$, the performance gain of using IS is rather limited and  $K^*$ is closer to $N$, i.e., IS policy is closer to RR. Deviation from RR becomes more apparent with increased heterogeneity in weights in which case $K^* \gg N$. The increased scov value $c$ leads to an increase in system AoI but also a slight increase in the IS pattern size.
When the heterogeneity is increased in the mean service times in Fig.~\ref{fig:IS1}(b), the IS pattern sizes are also larger and deviation from the RR policy becomes even more significant. In particular, for scenario B and when $r=50$,  $K^*=167$ (resp. $K^*=118$) when $c=1$ (resp. $c=0$), leading to around  
$20\%$ (resp. $19 \%$) reduction in the system AoI with IS in comparison to RR.  

\begin{figure*}[t]
  \centering
  \includegraphics[width=0.8\textwidth]{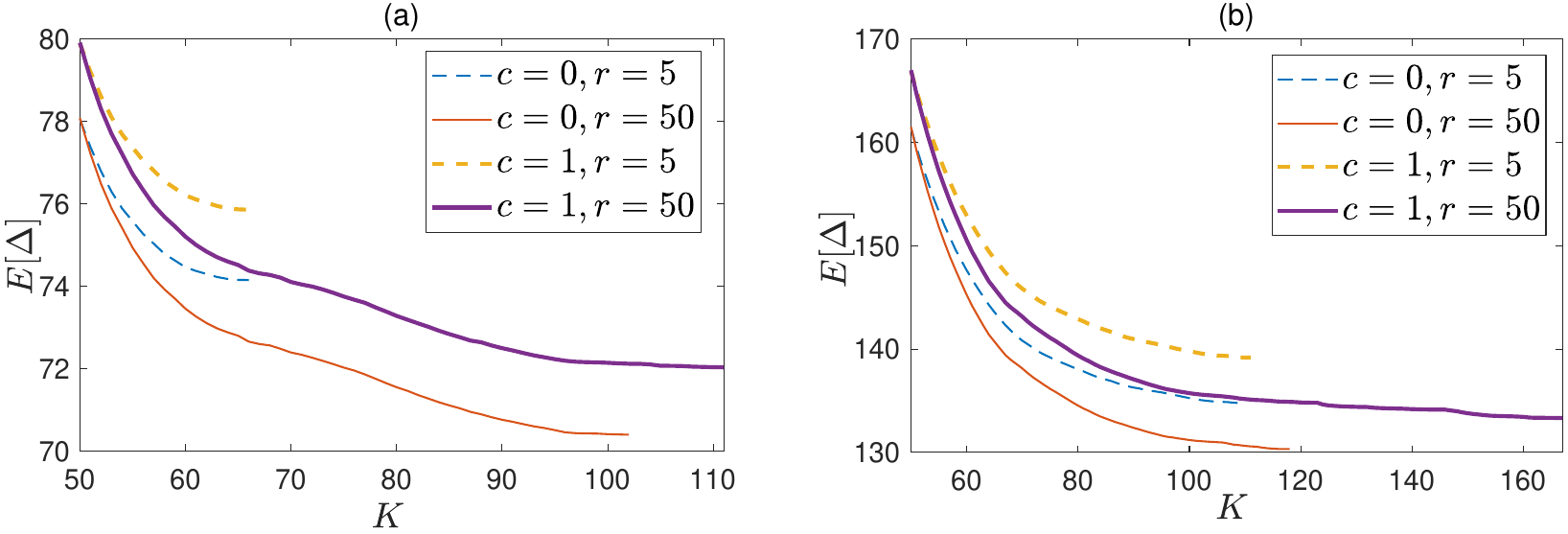}
  \caption{System AoI as a function of the iteration index $K$ in the IS-$1$ algorithm when $N=50$ and for two values of the fixed scov parameter $c$ and the ratio parameter $r$: (a) scenario A (b) scenario B.}
 \label{fig:IS1}
\end{figure*}

In the second numerical example, we compare IS with RR and two variations of another benchmark algorithm called SPMS (System PAoI Minimizing Scheduler). The first variation SPMS-1 is a P-GAW system for which the transmission probabilities are selected so as to minimize the system PAoI \cite{AKAR2023109668}.
The second variation SPMS-2 is proposed in \cite{akar_etal_infocom24} where a C-GAW system (as opposed to P-GAW) is used where the transmission frequencies of a user in the C-GAW pattern are approximations to the transmission probability $p_n$ found for SPMS-1 and a packet spreading algorithm (Algorithm~2 of \cite{akar_etal_infocom24}) is used to spread out as much as possible the multiple transmissions from a single source in the transmission pattern to reduce the system AoI. In this example, we vary the algorithm parameter $\varepsilon$ in Algorithm~1 of \cite{akar_etal_infocom24} from $0$ to $1$ in steps of $0.1$ to search over multiple pattern sizes to find the best pattern that gives the lowest weighted AoI among the ones that (almost) minimize the system PAoI. We use the same heterogeneous network model as in the previous example and fix $r=100$, $c=1$, but a more heterogeneous MCS allocation is used, i.e., $s_n \in \{ 1,4,16,64,256 \}$. The percentage gain in system AoI (resp. in system PAoI) with respect to the RR policy is depicted in Fig.~\ref{fig:IS2}(a) (resp. Fig.~\ref{fig:IS2}(b)) as a function of the number of sources $N$ for the three studied algorithms IS, SPMS-1, and SPMS-2. Although SPMS-1 presents the best performance in system PAoI, its system AoI performance may even be below that of the RR policy due to its probabilistic nature. However, the SPMS-2 policy which is derived from SPMS-1 is a cyclic policy and its system AoI performance is remarkable compared to RR and system PAoI performance is slightly below that of SPMS-1. The IS algorithm presented the best system AoI performance among all the three studied policies, whereas it is outperformed by SPMS-1 and SPMS-2, in terms of the system PAoI, which was not taken into consideration at all in the design of IS. Moreover, IS performance gain increased with the number of sources for system AoI. However, the high computational complexity of IS may hinder its use for large systems, e.g., hundreds of sources. To quantify this, we obtained the computation time of IS for this example while using MATLAB R2018 (using the {\tt timeit} Matlab function) on a 2.5 GHz Windows 10 Pro laptop with 8 GB of memory, which is reported in Table~\ref{tab:computation_time}.

\begin{table}[thb] 
 \caption{Computation time for IS with respect to $N$}
    \centering
    \begin{tabular}{|c|r|}
    \hline
        $N$ & Computation Time (s)\\ \hline \hline
        10 & 0.9632   \\ \hline
        20  &  16.5596\\ \hline
        30  & 96.1454\\ \hline
        40  & 323.8208\\ \hline
        50 & 863.0550\\ \hline 
    \end{tabular}   
    \label{tab:computation_time}
\end{table}

% Let $SAoI_{\text{X}}$ denote the system AoI attained by algorithm $\text{X}$ which refers to RR, IS, SPMS-1, or SPMS-2.
% Similarly, let $SPAoI_{\text{X}}$ denote the system PAoI attained by algorithm $\text{X}$.
% Subsequently, we define the two percentage gains below with respect to the RR policy:
% \begin{align}
%     G_{\text{X}}^{\text{AoI}} & = \frac{100(SAoI_{\text{RR}} - SAoI_{\text{X}})}{SAoI_{\text{RR}}}, \\ 
%     G_{\text{X}}^{\text{PAoI}} & = \frac{100(SPAoI_{\text{RR}} - SPAoI_{\text{X}})}{SPAoI_{\text{RR}}},
% \end{align}

\begin{figure*}[t]
  \centering
  \includegraphics[width=0.8\textwidth]{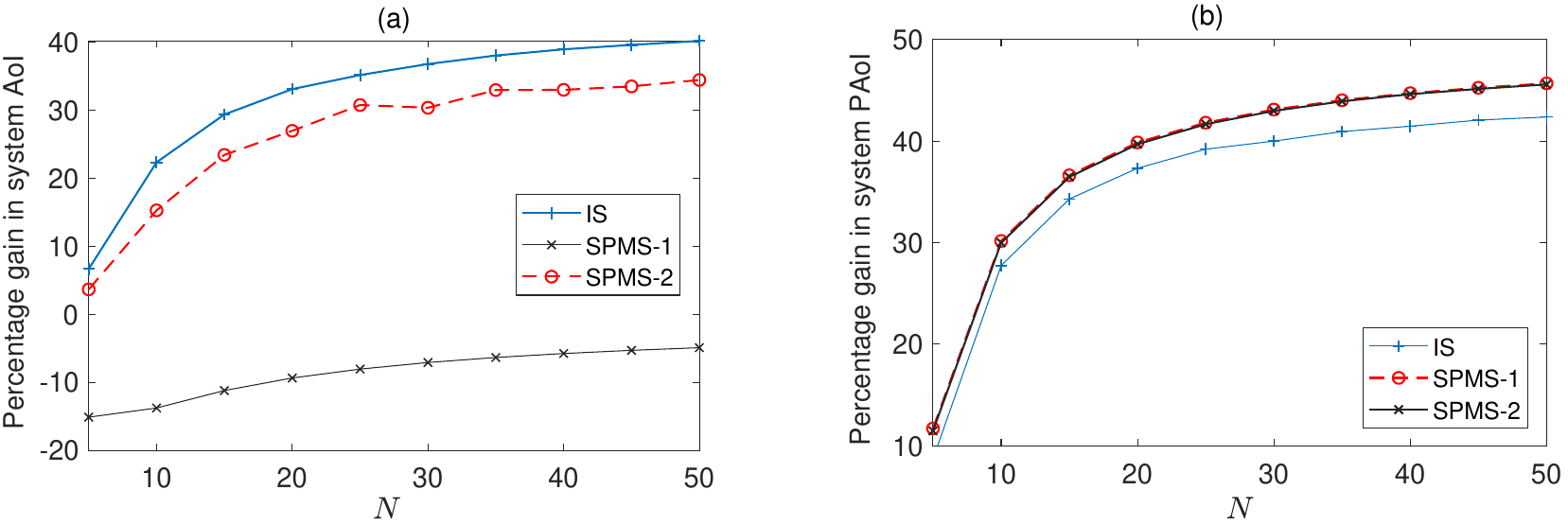}
  \caption{Percentage gain in (a) system AoI (b) system PAoI, achieved by the three studied algorithms, with respect to the RR policy as a function of the number of sources $N$.}
 \label{fig:IS2}
\end{figure*}

%\color{red}
\subsection{Comparison of IS with Age-Aware Scheduling}
In this subsection, we compare IS scheduling with an age-aware policy in two scenarios. Specifically, we focus on the Whittle-index policy which is proposed in \cite{kadota_tn18} for a time-slotted model with heterogeneous packet errors. In scenario A, we let $s_n=1, c_n = 0, w_n = n / \sum_{n'=1}^{N}n', \ 1 \leq n \leq N$, and assume error-free transmissions, which turns out to be a subcase of our proposed setting. Subsequently, we sweep the number of sources $N$ from 2 to 50 and plot the system AoI in Fig.~\ref{fig:Example_AgeAware}(a) where IS-$y$ denotes the IS algorithm using the algorithm parameter $y$. We observe that the Whittle-index policy slightly outperforms the IS policy whose performance improves with the choice of $y=2$ than $y=1$. Larger values of $y$ did not show substantial improvements and they are not reported. Moreover, as the number of sources increases, the performance gap between the IS-$2$ and Whittle-index policies slightly grows. However, this performance improvement with the Whittle-index policy is achieved at the expense of $O(N)$ runtime complexity  as opposed to $O(1)$ for IS.
In scenario B, we focus on a scenario with transmission errors where we let $\sigma_n$ denote the transmission success probability for source-$n$. To cope with packet errors in our setting, we assume the use of persistent ARQ re-transmissions which give rise to geometrically distributed service times for source-$n$ with parameter $\sigma_n$, i.e., $s_n = 1/\sigma_n, c_n = 1-\sigma_n$. Note that, there is no ARQ mechanism in the setting of \cite{kadota_tn18} where a fresh packet is generated (a new sample is taken) for every transmission irrespective of packet errors. We let $w_n \propto n$ as in the previous example and also assume the linearly spaced error model where $\sigma_n=\sigma_{n-1}+\delta, \ 2 \leq n \leq N, \  \delta>0$ with the ratio  $\frac{\sigma_N}{\sigma_1}$ fixed to 1.25 with $\sigma_N = 1$. Subsequently, we plot the system AoI as a function of the number of sources in Fig.~\ref{fig:Example_AgeAware}(b) where we observe that the Whittle-index policy slightly outperforms the IS policy similar to the previous example with outperformance increasing with the number of sources. However, an apparent drawback of the Whittle-index policy in this scenario is the generation of a new sample for each transmission, which can be costly for energy-constrained systems where there is an energy cost incurred for every sampling operation. To elaborate, we plot the average sampling rate (average number of generated samples per slot) with respect to the number of sources in Fig.~\ref{fig:Example_AgeAware_Sampling} where the average sampling rate appears to be around 9$\%$ lower for the  IS policy compared to the Whittle-index policy in this example.

\begin{figure*}[t]
	\centering
	\includegraphics[width=0.75\textwidth]{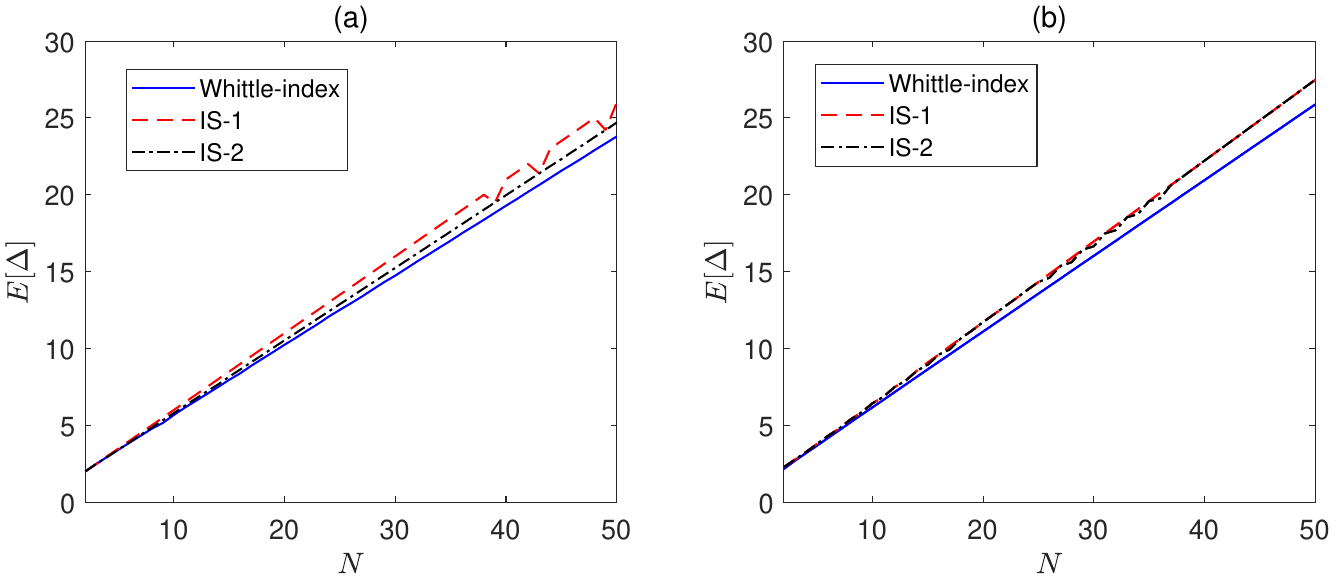}
	\caption{System AoI, obtained with the IS-$1$, IS-$2$ and Whittle-index \cite{kadota_tn18} policies, as a function of the number of sources $N$ for (a) scenario A (b) scenario B.}
	\label{fig:Example_AgeAware}
\end{figure*}

\begin{figure}[t]
	\centering
	\includegraphics[width=0.95\columnwidth]{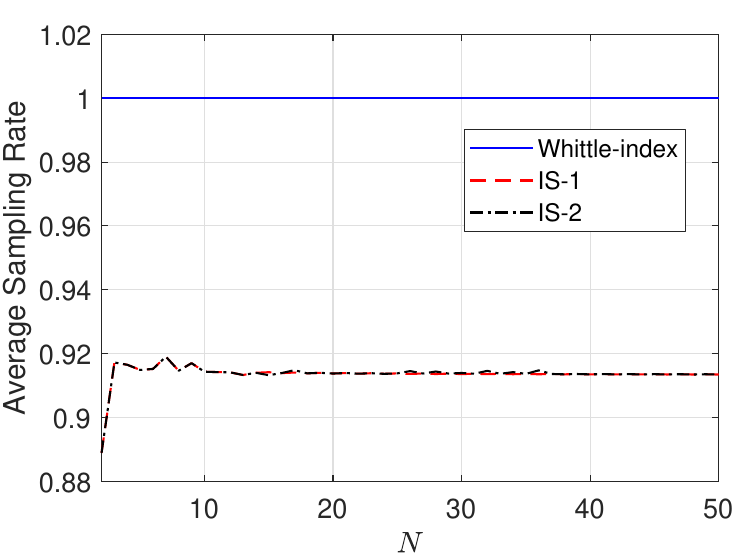}
	\caption{Average sampling rate, obtained with the IS-$1$, IS-$2$ and Whittle-index \cite{kadota_tn18} policies, as a function of the number of sources $N$ for scenario B.}
	\label{fig:Example_AgeAware_Sampling}
\end{figure}

%\color{black}

\section{Conclusions} \label{sec:conclusions}
In this paper, we study the benefits of cyclic scheduling (in terms of system AoI) with $O(1)$ runtime complexity for multi-source generate-at-will status update systems with heterogeneous user weights and service times, without packet errors. First, we present a numerically efficient analytical method for finding the mean of the AoI processes of the individual sources. This method is then used to derive the optimum cyclic scheduler for the special case of two sources. We conclude from several two-source examples that it is advantageous to employ a cyclic scheduler as opposed to a probabilistic or round robin scheduling, for system AoI minimization. 
% This optimum scheduler is used as a basis for a buffer management algorithm for a random arrival status update system with two sources. Numerical examples show that the proposed buffer management algorithm is very effective in reducing the system AoI when compared with conventional methods. 
For the case of status update systems involving  general number of information sources, a search-based heuristic scheduler, called IS (Insertion Search), is proposed whose effectiveness in system AoI reduction compared to other existing age-agnostic schemes is demonstrated through numerical examples. 
We also compared IS against the Whittle-index policy developed for slotted systems with heterogeneous error probabilities. We have shown that the performance of the IS policy is quite close to that of the Whittle-index policy especially for small number of sources whereas IS runtime complexity and sampling rate figures are lower.
What remains to be done is optimum scheduling for random arrivals for general number of sources and also development of new cyclic scheduling algorithms that can scale to thousands of information sources for massive connectivity scenarios. Moreover, taking non-zero packet error probabilities into consideration, is another interesting research direction. AoI scheduling with non-zero polling times, the case of remote monitors and the scheduler not being co-located, duty cycle constraints, and the employment of multiple servers as opposed to the single-server model of this paper, are also suggested as future work.

% \bibliographystyle{IEEEtran}
% \bibliography{AoIBibtex}

\begin{thebibliography}{10}
\providecommand{\url}[1]{#1}
\csname url@samestyle\endcsname
\providecommand{\newblock}{\relax}
\providecommand{\bibinfo}[2]{#2}
\providecommand{\BIBentrySTDinterwordspacing}{\spaceskip=0pt\relax}
\providecommand{\BIBentryALTinterwordstretchfactor}{4}
\providecommand{\BIBentryALTinterwordspacing}{\spaceskip=\fontdimen2\font plus
\BIBentryALTinterwordstretchfactor\fontdimen3\font minus
  \fontdimen4\font\relax}
\providecommand{\BIBforeignlanguage}[2]{{%
\expandafter\ifx\csname l@#1\endcsname\relax
\typeout{** WARNING: IEEEtran.bst: No hyphenation pattern has been}%
\typeout{** loaded for the language `#1'. Using the pattern for}%
\typeout{** the default language instead.}%
\else
\language=\csname l@#1\endcsname
\fi
#2}}
\providecommand{\BIBdecl}{\relax}
\BIBdecl

\bibitem{lts2015}
R.~D. Yates, ``Lazy is timely: Status updates by an energy harvesting source,''
  in \emph{IEEE ISIT}, June 2015.

\bibitem{kadota2021wifresh}
I.~Kadota, M.~S. Rahman, and E.~Modiano, ``Wifresh: Age-of-information from
  theory to implementation,'' in \emph{IEEE International Conference on
  Computer Communications and Networks (ICCCN)}, 2021.

\bibitem{chen2022age}
X.~Chen, K.~Gatsis, H.~Hassani, and S.~S. Bidokhti, ``Age of information in
  random access channels,'' \emph{IEEE Transactions on Information Theory},
  vol.~68, no.~10, pp. 6548--6568, 2022.

\bibitem{yu2024optimizing}
B.~Yu, Y.~Cai, D.~Wu, C.~Dong, R.~Zhang, and W.~Wu, ``Optimizing age of
  information for uplink cellular internet of things with random access,''
  \emph{IEEE Internet of Things Journal}, 2024, to appear.

\bibitem{ghose2018mac}
D.~Ghose, F.~Y. Li, and V.~Pla, ``{MAC} protocols for wake-up radio:
  Principles, modeling and performance analysis,'' \emph{IEEE Transactions on
  Industrial Informatics}, vol.~14, no.~5, pp. 2294--2306, 2018.

\bibitem{trotta2019bee}
A.~Trotta, M.~Di~Felice, L.~Bononi, E.~Natalizio, L.~Perilli, E.~F. Scarselli,
  T.~S. Cinotti, and R.~Canegallo, ``Bee-drones: Energy-efficient data
  collection on wake-up radio-based wireless sensor networks,'' in \emph{IEEE
  INFOCOM}, 2019, pp. 547--553.

\bibitem{shiraishi2022query}
J.~Shiraishi, A.~E. Kal{\o}r, F.~Chiariotti, I.~Leyva-Mayorga, P.~Popovski, and
  H.~Yomo, ``Query timing analysis for content-based wake-up realizing
  informative {IoT} data collection,'' \emph{IEEE Wireless Communications
  Letters}, vol.~12, no.~2, pp. 327--331, 2022.

\bibitem{shiraishi2024coexistence}
J.~Shiraishi, S.~Cavallero, S.~R. Pandey, F.~Saggese, and P.~Popovski,
  ``Coexistence of push wireless access with pull communication for
  content-based wake-up radios,'' \emph{arXiv preprint arXiv:2404.12816}, 2024.

\bibitem{mcs1}
R.~Kufakunesu, G.~P. Hancke, and A.~M. Abu-Mahfouz, ``A survey on adaptive data
  rate optimization in {LoRaWAN:} recent solutions and major challenges,''
  \emph{Sensors}, vol.~20, no.~18, 2020.

\bibitem{mcs2}
R.~Marini, K.~Mikhaylov, G.~Pasolini, and C.~Buratti, ``Low-power wide-area
  networks: Comparison of {LoRaWAN and NB-IoT} performance,'' \emph{IEEE
  Internet of Things Journal}, vol.~9, no.~21, pp. 21\,051--21\,063, 2022.

\bibitem{arq1}
A.~M. Cipriano, P.~Gagneur, G.~Vivier, and S.~Sezginer, ``Overview of {ARQ and
  HARQ} in beyond {3G} systems,'' in \emph{2010 IEEE 21st International
  Symposium on Personal, Indoor and Mobile Radio Communications Workshops},
  2010, pp. 424--429.

\bibitem{arq2}
A.~Ahmed, A.~Al-Dweik, Y.~Iraqi, H.~Mukhtar, M.~Naeem, and E.~Hossain, ``Hybrid
  automatic repeat request {(HARQ)} in wireless communications systems and
  standards: {A} contemporary survey,'' \emph{IEEE Communications Surveys \&
  Tutorials}, vol.~23, no.~4, pp. 2711--2752, 2021.

\bibitem{RoyYates__AgeOfInfo_Survey}
R.~D. Yates, Y.~Sun, D.~R. Brown, S.~K. Kaul, E.~Modiano, and S.~Ulukus, ``Age
  of information: An introduction and survey,'' \emph{IEEE Journal on Selected
  Areas in Communications}, vol.~39, no.~5, pp. 1183--1210, May 2021.

\bibitem{kosta_etal_survey}
A.~Kosta, N.~Pappas, and V.~Angelakis, ``Age of information: A new concept,
  metric, and tool,'' \emph{Foundations and Trends in Networking}, vol.~12,
  no.~3, pp. 162--259, 2017.

\bibitem{kadota_tn18}
I.~Kadota, A.~Sinha, E.~Uysal-Biyikoglu, R.~Singh, and E.~Modiano, ``Scheduling
  policies for minimizing age of information in broadcast wireless networks,''
  \emph{IEEE/ACM Transactions on Networking}, vol.~26, no.~6, pp. 2637--2650,
  2018.

\bibitem{he_yuan_ephremides_T-IT18}
Q.~He, D.~Yuan, and A.~Ephremides, ``Optimal link scheduling for age
  minimization in wireless systems,'' \emph{IEEE Transactions on Information
  Theory}, vol.~64, no.~7, pp. 5381--5394, 2018.

\bibitem{maatouk_etal_TWC21}
A.~Maatouk, S.~Kriouile, M.~Assad, and A.~Ephremides, ``On the optimality of
  the {W}hittleâ€™s index policy for minimizing the age of information,''
  \emph{IEEE Transactions on Wireless Communications}, vol.~20, no.~2, pp.
  1263--1277, 2021.

\bibitem{kriouile_etal_T-IT22}
S.~Kriouile, M.~Assaad, and A.~Maatouk, ``On the global optimality of
  {W}hittleâ€™s index policy for minimizing the age of information,''
  \emph{IEEE Transactions on Information Theory}, vol.~68, no.~1, pp. 572--600,
  2022.

\bibitem{eywa}
C.~Li, S.~Li, Q.~Liu, Y.~T. Hou, W.~Lou, and S.~Kompella, ``Eywa: A general
  approach for scheduler design in {AoI} optimization,'' in \emph{IEEE
  Infocom}, May 2023.

\bibitem{gau_ciss24}
E.~O. Gamgam, N.~Akar, and S.~Ulukus, ``Cyclic scheduling for age of
  information minimization with generate at will status updates,'' in
  \emph{Conference on Information Sciences and Systems (CISS)}, 2024.

\bibitem{akar_gamgam_comlet23}
N.~Akar and E.~O. Gamgam, ``Distribution of age of information in status update
  systems with heterogeneous information sources: An absorbing {Markov}
  chain-based approach,'' \emph{IEEE Communications Letters}, vol.~27, no.~8,
  pp. 2024--2028, 2023.

\bibitem{kaul_etal_infocom12}
S.~Kaul, R.~D. Yates, and M.~Gruteser, ``Real-time status: How often should one
  update?'' in \emph{IEEE Infocom}, March 2012.

\bibitem{Yates__SBR}
R.~D. Yates and S.~K. Kaul, ``The age of information: Real-time status updating
  by multiple sources,'' \emph{IEEE Transactions on Information Theory},
  vol.~65, no.~3, pp. 1807--1827, March 2019.

\bibitem{najm2018status}
E.~Najm and E.~Telatar, ``Status updates in a multi-stream {M/G/1/1} preemptive
  queue,'' in \emph{IEEE Infocom}, April 2018.

\bibitem{yates_mgf}
R.~D. {Yates}, ``The age of information in networks: Moments, distributions,
  and sampling,'' \emph{IEEE Transactions on Information Theory}, vol.~66,
  no.~9, pp. 5712--5728, 2020.

\bibitem{elmagid_dhillon_tit22}
M.~A. Abd-Elmagid and H.~S. Dhillon, ``Closed-form characterization of the
  {MGF} of {AoI} in energy harvesting status update systems,'' \emph{IEEE
  Transactions on Information Theory}, vol.~68, no.~6, pp. 3896--3919, 2022.

\bibitem{moltafet_etal_tcom22}
M.~Moltafet, M.~Leinonen, and M.~Codreanu, ``Moment generating function of age
  of information in multisource {M/G/1/1} queueing systems,'' \emph{IEEE
  Transactions on Communications}, vol.~70, no.~10, pp. 6503--6516, 2022.

\bibitem{maatuk_isit22}
A.~Maatouk, M.~Assaad, and A.~Ephremides, ``Analysis of an age-dependent
  stochastic hybrid system,'' in \emph{IEEE ISIT}, June 2022.

\bibitem{bedewy_etal_tit21}
A.~M. Bedewy, Y.~Sun, S.~Kompella, and N.~B. Shroff, ``Optimal sampling and
  scheduling for timely status updates in multi-source networks,'' \emph{IEEE
  Transactions on Information Theory}, vol.~67, no.~6, pp. 4019--4034, 2021.

\bibitem{joo_eryilmaz_ton18}
C.~Joo and A.~Eryilmaz, ``Wireless scheduling for information freshness and
  synchrony: Drift-based design and heavy-traffic analysis,'' \emph{IEEE/ACM
  Transactions on Networking}, vol.~26, no.~6, pp. 2556--2568, Dec. 2018.

\bibitem{kadota_tmc21}
I.~Kadota and E.~Modiano, ``Minimizing the age of information in wireless
  networks with stochastic arrivals,'' \emph{IEEE Transactions on Mobile
  Computing}, vol.~20, no.~3, pp. 1173--1185, 2021.

\bibitem{yang_etal_icassp20}
H.~H. Yang, A.~Arafa, T.~Q.~S. Quek, and H.~Vincent~Poor, ``Age-based
  scheduling policy for federated learning in mobile edge networks,'' in
  \emph{ICASSP}, 2020, pp. 8743--8747.

\bibitem{timely_throughput_1}
I.~Kadota, A.~Sinha, and E.~Modiano, ``Scheduling algorithms for optimizing age
  of information in wireless networks with throughput constraints,''
  \emph{IEEE/ACM Transactions on Networking}, vol.~27, no.~4, pp. 1359--1372,
  2019.

\bibitem{timely_throughput_2}
E.~Fountoulakis, T.~Charalambous, A.~Ephremides, and N.~Pappas, ``Scheduling
  policies for {AoI} minimization with timely throughput constraints,''
  \emph{IEEE Transactions on Communications}, vol.~71, no.~7, pp. 3905--3917,
  2023.

\bibitem{energy_constraint}
K.~Saurav and R.~Vaze, ``Scheduling to minimize age of information with
  multiple sources,'' \emph{IEEE Journal on Selected Areas in Information
  Theory}, vol.~4, pp. 539--550, 2023.

\bibitem{csma}
V.~Tripathi, N.~Jones, and E.~Modiano, ``{Fresh-CSMA: A} distributed protocol
  for minimizing age of information,'' in \emph{IEEE Infocom}, May 2023.

\bibitem{spatially_distributed}
N.~Jones and E.~Modiano, ``Minimizing age of information in spatially
  distributed random access wireless networks,'' in \emph{IEEE Infocom}, May
  2023.

\bibitem{eywa_phDthesis}
C.~Li, ``Optimizing information freshness in wireless networks,'' Ph.D.
  dissertation, Virginia Polytechnic Institute and State University, 2023.

\bibitem{AKAR2023109668}
N.~Akar and E.~Karasan, ``Is proportional fair scheduling suitable for
  age-sensitive traffic?'' \emph{Computer Networks}, vol. 226, p. 109668, 2023.

\bibitem{akar_etal_infocom24}
N.~Akar, S.~Liyanaarachchi, and S.~Ulukus, ``Scalable cyclic schedulers for age
  of information optimization in large-scale status update systems,'' in
  \emph{Age and Semantics of Information Workshop (ASoI), IEEE Infocom}, 2024.

\bibitem{kleinrock_book}
L.~Kleinrock, \emph{Queueing Systems. Volume 1: Theory}.\hskip 1em plus 0.5em
  minus 0.4em\relax Wiley-Interscience, 1975.

\bibitem{boyd2004convex}
S.~P. Boyd and L.~Vandenberghe, \emph{Convex Optimization}.\hskip 1em plus
  0.5em minus 0.4em\relax Cambridge University Press, 2004.

\end{thebibliography}

\end{document}